\documentclass[12pt]{iopart}

\usepackage{iopams}  

\usepackage[utf8]{inputenc}
\usepackage{amssymb}
\expandafter\let\csname equation*\endcsname\relax
\expandafter\let\csname endequation*\endcsname\relax

\usepackage{bbold}
\usepackage{amsmath}

\usepackage{slashed}
\usepackage{nicefrac}

\usepackage{braket}

\usepackage{hyperref}
\usepackage{cite}

\usepackage{color}
\usepackage{amsthm}
\usepackage{amsfonts}
\usepackage{shuffle}
\usepackage{cleveref}
\newcommand{\rd}{\mathrm{d}}
\newcommand{\cF}{\mathcal{F}}
\newcommand{\cU}{\mathcal{U}}
\newcommand{\eps}{\epsilon}
\newcommand{\cN}{\mathcal{N}}
\newcommand{\cI}{\mathcal{I}}

\def\beq{\begin{equation}}
\def\eeq{\end{equation}}
\DeclareMathOperator*{\res}{Res}
\DeclareMathOperator*{\dlog}{\rd \! \log}

\def\bra#1{\left\langle #1\right|}

\def\sqbra#1{\left[ #1\right|}
\def\sqket#1{\left| #1\right]}
\def\gb #1{ \left\langle #1 \right]}
\def\tgb #1{ \left[ #1 \right\rangle}

\newtheorem{prop}{Proposition}
\newtheorem{conj}{Conjecture}

\newtheorem{cor}{Corollary}
\newtheorem{ex}{Example}
\theoremstyle{remark}
\newtheorem{rem}{Remark}

\def\vev#1{\left\langle #1 \right\rangle}

\usepackage{etoolbox}

\makeatletter
\def\@mkboth#1#2{}
\newlength\appendixwidth
\preto\appendix{\addtocontents{toc}{\protect\patchl@section}}
\newcommand{\patchl@section}{%
  \settowidth{\appendixwidth}{\textbf{Appendix }}%
  \addtolength{\appendixwidth}{1.5em}%
  \patchcmd{\l@section}{1.5em}{\appendixwidth}{}{\ddt}%
}
\makeatother

\begin{document}

\begin{flushright}
SAGEX-22-04, BONN-TH-2022-03, CERN-TH-2022-021
\end{flushright}

\title[Feynman integrals]{The SAGEX Review on Scattering Amplitudes\\ 
\vskip 0.4cm Chapter 3: Mathematical Structures in Feynman Integrals}

\author{Samuel Abreu$^{1,2}$, Ruth Britto$^{3}$ and Claude Duhr$^4$}

 \address{$^1$ Theoretical Physics Department, CERN, 1211 Geneva, Switzerland}
 \address{$^2$ Higgs Centre for Theoretical Physics, School of Physics and Astronomy,
 The University of Edinburgh, Edinburgh EH9 3FD, Scotland, UK}
 \address{$^3$ School of Mathematics and Hamilton Mathematics Institute, Trinity College, Dublin 2, Ireland}
  \address{$^4$ Bethe Center for Theoretical Physics, Universit\"at Bonn, D-53115, Germany}

\ead{samuel.abreu@cern.ch, brittor@tcd.ie, cduhr@uni-bonn.de}

\begin{abstract}
Dimensionally-regulated Feynman integrals are a cornerstone of all perturbative computations in quantum field theory. They are known to exhibit a rich mathematical structure, which has led to the development of powerful new techniques for their computation. We review some of the most recent advances in our understanding of the analytic structure of multiloop Feynman integrals in dimensional regularisation. In particular, we give an overview of modern approaches to computing Feynman integrals using differential equations, and we discuss some of the properties of the functions that appear in the solutions. We then review how dimensional regularisation has a natural mathematical interpretation in terms of the theory of twisted cohomology groups, and how many of the well-known ideas about Feynman integrals arise naturally in this context. This is Chapter 3 of a series of review articles on scattering amplitudes, of which Chapter 0 \cite{Travaglini:2022uwo} presents an overview and Chapter 4 \cite{Blumlein:2022zkr} contains closely related topics.
\end{abstract}

%
%
%
\maketitle
%
%


\tableofcontents

\newpage


%
%

\section{Background and definitions}

An $L$-loop Feynman integral is an integral of the form:
\beq\label{eq:fint}
I(p_1,\ldots,p_E;m_1^2,\ldots,m_p^2;\nu;D) = 
\int\left(\prod_{j=1}^Le^{\gamma_E\eps}\frac{\rd^Dk_j}{i\pi^{D/2}}\right)
\frac{\cN(\{k_j\cdot k_l, k_j\cdot p_l\};D)}{\prod_{j=1}^p(m_j^2-q_j^2-i\varepsilon)^{\nu_j}}\,.
\eeq
Here $\gamma_E=-\Gamma'(1)$ is the Euler-Mascheroni constant,  $\nu = (\nu_1,\ldots,\nu_p)\in\mathbb{Z}^p$, 
and we work in dimensional regularisation with $D=D_0-2\eps$ dimensions and $D_0\in\mathbb{N}$.
Depending on the context,
we will use $D$ or $\epsilon$ interchangeably to denote quantities that depend on the dimensional regulator.
The factors $1/(m_j^2-q_j^2-i\varepsilon)$ are called \emph{propagators}, and
we use the usual Feynman $i\varepsilon$-prescription to deform the integration contour away from the propagator
poles. Hence note the distinction between the dimensional regularisation parameter $\eps$ and the infinitesimal $\varepsilon$ for the Feynman prescription for the propagator.
The integral is a function of the propagator masses $m_1^2,\ldots,m_p^2$ and of the external momenta  
$p_1,\ldots,p_E$, constrained by momentum  conservation $\sum_{j=1}^Ep_j=0$
(we assume without loss of generality that all
external momenta are incoming).
The squared propagator masses are assumed to be positive, and the external momenta are real Minkowski momenta which 
we also assume to be $D$-dimensional (we assume that, aside from momentum conservation,
they are all linearly independent).
The numerator $\cN$ is a polynomial in the dot products involving at least one of the loop momenta $k_i$.
The momenta $q_i$ flowing through the propagators have the form
\beq\label{eq:propmomenta}
q_i = \sum_{j=1}^L\alpha_{ij}k_j + \sum_{j=1}^{E-1}\beta_{ij}p_j\,,\qquad \alpha_{ij},\beta_{ij}\in\{-1,0,+1\}\,.
\eeq
We define $|\nu|$ as the sum of the exponents of the denominators,
\begin{equation}
	|\nu|=\sum_{j=1}^p\nu_j\,.
\end{equation}

The Feynman integral in \cref{eq:fint} is invariant under Lorentz transformations $\Lambda$ in $D$ dimensions:
\beq\label{eq:Lorentz}
I(\Lambda p_1,\ldots,\Lambda p_E;m_1^2,\ldots,m_p^2;\nu;D) = I(p_1,\ldots,p_E;m_1^2,\ldots,m_p^2;\nu;D)\,.
\eeq
As a consequence, the integral only depends on the \emph{external scales}
\begin{equation}\label{eq:extScales}
x=(\{p_i\cdot p_j\}_{1\le i,j\le E},\{m_j^2\}_{1\le j\le p})\,.
\end{equation}
We will often write $I(x; \nu;D)$ instead of $I(p_1,\ldots,p_E;m_1^2,\ldots,m_p^2;\nu;D)$, 
in order to make the Lorentz invariance manifest. We note that because of momentum 
conservation not all $\{p_i\cdot p_j\}_{1\le i,j\le E}$ are independent, so it should be understood
that $x$ contains only an independent set of such dot products.

\begin{rem} It is possible to relax the condition $\nu_i\in\mathbb{Z}$. 
This leads to the notion of \emph{generic} Feynman integrals considered in ref.~\cite{speer1969,speerwestwater}.
In particular, for certain classes of Feynman integrals, an $L$-loop Feynman integral
can be written as an $(L-1)$-loop generic Feynman integral.
\end{rem}

\subsection{Some properties of Feynman integrals}

\paragraph{Feynman graphs.} Feynman integrals as presented in \cref{eq:fint} are associated to \emph{Feynman graphs} with $L$ loops, where 
for every propagator there is an internal edge $e_j$ labelled by $(q_j,m_j^2,\nu_j)$, 
and momentum must be conserved at every vertex. 
The external edges are labelled by the inflowing external momenta.

\paragraph{Mass dimension of a Feynman integral.} 
Consider the action on Feynman integrals of a rescaling of the external momenta and the propagator masses, 
$(p_j,m_j) \to (\lambda\,p_j,\lambda\,m_j)$, $\lambda\in\mathbb{R}^{\ast} = \mathbb{R}\setminus\{0\}$.
It is straightforward to check that all external scales in \cref{eq:extScales} rescale like 
$x\to \lambda^2\, x$. Assuming 
that the numerator is homogeneous (as is always the case in physics applications), 
\begin{equation}
\cN(\lambda^2\,\{k_j\cdot k_l, k_j\cdot p_l\};D) = 
\lambda^{\alpha_{\cN}}\,\cN(\{k_j\cdot k_l, k_j\cdot p_l\};D)\,,
\end{equation}
where $\alpha_{\cN} = [\cN]$ is the \emph{mass dimension} of the numerator, we have
\beq\label{eq:Rescaling}
I(\lambda^2 x;\nu;D) = \lambda^{\alpha_I} \,I( x;\nu;D) \,,
\eeq
where the mass dimension of the integral is
\beq\label{eq:energy_dimension}
[I( x;\nu;D)] = \alpha_{I} = [\cN] + LD-2\,|\nu|\,.
\eeq
Note that  $[I( x;\nu;D)] = m-2L\eps$  for some integer $m\in \mathbb{Z}$,
which implies that the mass dimension of an integral is never zero in dimensional regularisation.

\paragraph{Symmetries.} The behaviour under Lorentz transformations and rescalings, see 
\cref{eq:Lorentz,eq:Rescaling}, can be recast in terms of the infinitesimal action of 
the generators of the Lorentz transformations and dilatations:
\begin{align}
\label{eq:LI}L_{\mu\nu} I(p_1,\ldots,p_E;m_1^2,\ldots,m_p^2;\nu;D) &\,= 0\,, \\
\label{eq:dilatations}\mathcal{D} I(p_1,\ldots,p_E;m_1^2,\ldots,m_p^2;\nu;D) &\,= 
\alpha_I\,I(p_1,\ldots,p_E;m_1^2,\ldots,m_p^2;\nu;D)\,,
\end{align}
with
\begin{align}\begin{split}\label{eq:dilatations_generator}
L_{\mu\nu} = \sum_{i=1}^E\left(p_{i,\mu}\frac{\partial}{\partial p_i^{\nu}}-p_{i,\nu}\frac{\partial}{\partial p_i^{\mu}}\right)\,,\\
\mathcal{D}= \sum_{i=1}^E p_i^{\mu}\frac{\partial}{\partial p_i^{\mu}} + \sum_{i=1}^p m_i\frac{\partial}{\partial m_i}\,.
\end{split}\end{align}

\paragraph{Scaleless integrals.} A Feynman integral is said to be \emph{scaleless} if $x=\vec 0$. 
This means that the integral does not depend on any external scale. 
\begin{prop}
In dimensional regularisation, all scaleless integrals vanish.
\end{prop}
\begin{proof}
Since there is no external scale, the integral is a constant. On the other hand, we must have
\beq
I(\lambda^2 x; \nu;D) = \lambda^{\alpha_I} \,I( x; \nu;D) \,,\qquad 
\textrm{for all $\lambda\in\mathbb{R}^{\ast}$}\,.
\eeq
Since $\alpha_I\neq0$ in dimensional regularisation---see \cref{eq:energy_dimension}---the
previous equation can only be consistent if $I(x;\nu;D)=0$.
\end{proof}
\begin{cor}\label{cor:neg_powers}
If $\nu_j\le 0$ for all $1\le j\le p$, then $I(x;\nu;D)=0$ in dimensional regularisation.
\end{cor}
\begin{proof}
If $\nu_j\le 0$ for all $1\le j\le p$, then the integrand is a polynomial. 
We can then write the integral as a linear combination of scaleless integrals.
\end{proof}

\subsection{Parametric representations}\label{sec:paramRep}

The representation of Feynman integrals in \cref{eq:fint} is the one that is the most
straightforwardly associated with Feynman graphs and Feynman rules.
However, it is not the most convenient to evaluate the integrals and it
obscures some of their properties. In this subsection we present alternative
representations that can be helpful in addressing these questions, mostly focusing
on the Feynman integrals with unit numerator, $\cN=1$. The parametric representations
of the integrals that will be discussed here can also be used as the \emph{definition}
of a Feynman integral in dimensional regularisation, where $D$ is simply another
parameter the integrals depend on, instead of having the physical interpretation
of being the non-integer dimensional space in which the momenta live.

We will not discuss all existing parametric representations
for Feynman integrals, but instead restrict ourselves to the ones that will 
feature in the rest of this review. We refer the reader to other references
for a more detailed discussion, e.g.,~refs.~\cite{Smirnov:2012gma,Panzer:2015ida,Weinzierl:2022eaz}.

\paragraph{Schwinger parametrisation.} The Schwinger-parameter representation is obtained by using:
\begin{equation}
	\frac{1}{X^\nu}=\frac{1}{\Gamma(\nu)}\int_0^\infty \rd\alpha\,\alpha^{\nu-1}e^{-\alpha X}\,,
	\qquad \text{for } X>0\,,\quad \text{Re}(\nu)>0\,,
\end{equation}
to rewrite each of the factors in the denominator of \cref{eq:fint}. Through
standard manipulations (including Wick rotation of the energy component of the loop momenta), 
we can then integrate over the loop momenta and obtain:
\begin{equation}\label{eq:Schwinger}
I(x; \nu;D) = e^{\gamma_E L\eps}\,
\prod_{j=1}^p\int_0^\infty\rd\alpha_j\,\frac{\alpha_j^{\nu_j-1}}{\Gamma(\nu_j)}\,
\cU(\alpha)^{-D/2}\,e^{-\cF(\alpha;x)/\cU(\alpha)}\,.
\end{equation}
Here, the $\alpha_j$ are called Schwinger parameters, and
$\cU(\alpha)$ and $\cF(\alpha;x)$ are polynomials that are associated with the corresponding
Feynman graph (i.e., they can be determined directly from the Feynman graph
associated with the Feynman integral). Their precise form is not important
for us here, but we highlight some of their properties that will be used later:
\begin{itemize}
	\item $\cU(\alpha)$ is the determinant of an $L\times L$ matrix whose entries only depend
	on the Schwinger parameters $\alpha_j$. It is an homogeneous polynomial of degree $L$
	in the $\alpha_j$.
	\item $\cF(\alpha;x)$ depends on both the Schwinger parameters $\alpha_j$ and the external scales $x$.
	It is an homogeneous polynomial of degree $L+1$ in the $\alpha_j$ of mass dimension 2.
	It has the form:
	\begin{equation}\label{eq:FmassDep}
		\cF(\alpha;x)=\cU(\alpha)\sum_{j=1}^pm_j^2\,\alpha_j+\widetilde{\cF}(\alpha;x)\,,
	\end{equation}
	where $\widetilde{\cF}(\alpha;x)$ is independent of all internal masses $m_j^2$.
\end{itemize}
The polynomials $\cU$ and $\widetilde{\cF}$ are also known as the first and second Symanzik polynomials.

\paragraph{Feynman parametrisation.} This is perhaps the most well known parametric representation, 
and can be derived either directly from \cref{eq:fint},
or from the Schwinger-parameter representation.
Starting from \cref{eq:Schwinger}, we insert
\begin{equation}
	1=\int_0^\infty \rd u\,\delta\left(u-\sum_{j=1}^p\alpha_j\right)
\end{equation}
and after some manipulation we obtain
\begin{align}\begin{split}\label{eq:Feynman}
I( x; \nu;D) =\, & e^{\gamma_E L\eps}\,\Gamma\left(|\nu|-\frac{LD}{2}\right)\,
\prod_{j=1}^p\int_0^\infty\rd\alpha_j\,\frac{\alpha_j^{\nu_j-1}}{\Gamma(\nu_j)}\,
\delta\left(1-\sum_{j=1}^{p}\alpha_j\right)\,\,\frac{\cU( \alpha)^{|\nu|-\frac{(L+1)D}{2}}}
{\cF( \alpha; x)^{|\nu|-\frac{LD}{2}}}\,,
\end{split}\end{align}
where the $\alpha_j$ are now called Feynman parameters.
The Feynman-parameter representation in \cref{eq:Feynman} is not as general
as it could be. Indeed, it can be shown that it is in fact a projective integral
over a simplex in real projective space of dimension $p-1$, 
and \cref{eq:Feynman} is only a particular realisation of such an integral. 
Other realisations can, for instance, be obtained by restricting the sum in the delta 
function to be over only a subset of the Feynman parameters (usually
referred to as the \emph{Cheng-Wu theorem} in physics~\cite{Cheng:1987ga}), or even any other linear combination thereof~\cite{Panzer:2015ida}. 

The Feynman-parameter representation has been widely used to compute Feynman integrals
by direct integration, both numerically and analytically. It also allows us to make an observation on
Feynman integrals that is completely obscured in the momentum-space representation 
of \cref{eq:fint} and not as clear in the Schwinger parameter representation
of \cref{eq:Schwinger}: we see that the powers of the denominators $\nu_i$ and the dimension
$D$ appear in a very similar way in the integrand of \cref{eq:Feynman}, namely they are both in the exponents
of the $\cU(\alpha)$ and $\cF(\alpha;x)$ polynomials. 
We will see later that we can find relations between integrals with different
values of the $\nu_i$, but also with different values of $D$.

\paragraph{Cutkosky-Baikov parametrisation.} The final parametric representation we discuss is the Cutkosky-Baikov representation.
We start by noting that the integrand of the Feynman integral in \cref{eq:fint}
depends on $K=L+E-1$ momenta: the $L$ loop momenta $k_1,\ldots,k_L$ and the $E-1$ 
independent external momenta, say $p_1,\ldots, p_{E-1}$. Let us denote these momenta as $t_1,\ldots,t_K$, 
in the order specified above. We then consider all dot products $\tau_{ij}=t_i\cdot t_j$
for $1\leq i\leq L$ and $1\leq j\leq K$ that involve at least one loop momentum.
It is straightforward to check that there are $N=L(L+1)/2+L(E-1)$ such dot products.
We can change variables from the components of the loop momenta to the $\tau_{ij}$
to find:
\begin{align}\begin{split}
	I( x;\nu;D)=\,&
	\frac{(-1)^L\,\pi^{\frac{L-N}{2}}e^{\gamma_E L\eps}\left[G(p_1,\ldots,p_{E-1})\right]^{\frac{E-D}{2}}}
	{\Gamma\left(\frac{D-K+1}{2}\right)\Gamma\left(\frac{D-K+2}{2}\right)\ldots\Gamma\left(\frac{D-E+1}{2}\right)}
	\int_\Delta \prod_{i=1}^L\prod_{j=i}^K \rd\tau_{ij}\\
	&\left[G(t_1,\ldots,t_K)\right]^{\frac{D-K-1}{2}}
	\frac{\cN(\{\tau_{kl}; x\};D)}{\prod_{a=1}^p(m_a^2-q_a^2-i\varepsilon)^{\nu_a}}\,,
\end{split}\end{align}
where $G(a_1,\ldots,a_n)$ is the Gram determinant of the momenta $a_1$ through $a_n$. 
The integration domain $\Delta$ is bounded by the surface $G(t_1,\ldots,t_K)=0$.
We next note that the inverse propagators $z_a=(m_a^2-q_a^2)$ are linear in the $\tau_{ij}$, see \cref{eq:propmomenta}.
Provided there is an invertible transformation from the $z_a$ to the $\tau_{ij}$, we can then trivially
change variables from the $\tau_{ij}$ to the $z_a$. It is very common that there are fewer propagators than 
the number $N$ of $\tau_{ij}$. This situation can be dealt with simply  by considering a larger class of integrals
with extra propagators, so that the transformation from the $z_a$ to the $\tau_{ij}$ is invertible
(this defines a complete family of Feynman integrals, see \cref{sec:ibprels} below), 
and then set the power $\nu_i$ associated with the extra propagators to zero. 
Assuming there is an invertible transformation between the $z_a$ and the $\tau_{ij}$, we can write:
\begin{align}\begin{split}
\label{eq:Baikov}
	I( x;\nu;D)=\,&
	\frac{(-1)^L\,\pi^{\frac{L-N}{2}}e^{\gamma_E L\eps}\left[G(p_1,\ldots,p_{E-1})\right]^{\frac{E-D}{2}}}
	{\Gamma\left(\frac{D-K+1}{2}\right)\Gamma\left(\frac{D-K+2}{2}\right)\ldots\Gamma\left(\frac{D-E+1}{2}\right)}
	\int_\Delta \prod_{a=1}^p\rd z_a\,\mathcal{B}(z)^{\frac{D-K-1}{2}}
	\frac{\cN(\{z_k; x\};D)}{\prod_{c=1}^pz_c^{\nu_c}}\,,
\end{split}\end{align}
where $\mathcal{B}(z)$ is the Baikov polynomial associated with the Feynman integral $I( x;\nu;D)$
(strictly speaking, with the complete family of Feynman integrals defined by $I( x;\nu;D)$):
\begin{equation}\label{eq:baikovPol}
	\mathcal{B}(z)=G(t_1,\ldots,t_K)\,.
\end{equation}

\begin{rem} (Cut Feynman Integrals)\label{rem:cutFeynmanIntegrals}
The Cutkosky-Baikov representation in \cref{eq:Baikov} is often not the most practical for 
computing $I(x;\nu;D)$. However, this representation is well suited for studying
certain properties of Feynman integrals. For instance, it is a natural starting point
for defining and studying so-called \emph{cut Feynman integrals},
where a subset of the inverse propagators are set to zero.
(The representation in \cref{eq:Baikov} was in fact first introduced in ref.~\cite{Cutkosky:1960sp}
precisely to study the Landau singularities of Feynman integrals \cite{Landau:1959fi}, which
are closely related to their cuts; see also refs.~\cite{Harley:2017qut,Bosma:2017ens,Frellesvig:2017aai}
for discussions about computing cut integrals from the Cutkosky-Baikov representation.) 
In the representation of \cref{eq:Baikov}, this cut condition can be imposed by taking residues where 
the corresponding $z_a$ are zero. We note that there is a close connection between evaluating residues 
and choosing a contour that encircles the associated poles, so that cut 
Feynman integrals correspond to Feynman integrals evaluated on modified integration contours 
(see e.g.~ref.~\cite{Abreu:2017ptx} for a discussion in the context of one-loop integrals).
Of particular interest are contours that encircle the poles associated with all active propagators (those whose $\nu_i$ are positive), corresponding to a \emph{maximal cut} of the integral.
Beyond one loop, maximal cuts might not be unique, and there may be inequivalent choices of how to integrate over the variables that are not localised by the maximal-cut conditions.
\end{rem}

Given the relation between cut Feynman integrals and residues, 
it follows that:
\begin{prop}\label{prop:cuts}
	If $\nu_i\leq 0$ for any of the cut propagators of a cut Feynman 
	integral, then the cut integral vanishes.
\end{prop}


\section{Linear relations among Feynman integrals}
\label{sec:ibps}

Instead of seeing the Feynman integral as a function of the external scales $x$ 
for fixed values of the propagator exponents $\nu$ and the space-time dimension $D$, 
we can also interpret it as a function of $\hat{\nu}=(\nu_0,\nu_1,\ldots,\nu_p)\in\mathbb{Z}^{p+1}$,
with $\nu_0=D_0/2$, for fixed $x$. In other words, if we fix the external scales $x$,
we can associate to every point $\hat{\nu}$ on the lattice $\mathbb{Z}^{p+1}$ the 
Feynman integral $I(x;\nu;D)$. In this way we obtain an infinite \emph{family of Feynman integrals} 
labelled by the lattice points. Note that Corollary~\ref{cor:neg_powers} implies that the integral 
vanishes unless at least one of the $\nu_1,\ldots,\nu_p$ is strictly positive.  
The point of this section is to show that these integrals are not independent, but there are linear 
relations among integrals for different lattice points, and one can always identify a finite basis 
of integrals that generates this infinite family. We start by discussing the linear relations 
relating integrals with different values of $\nu\in\mathbb{Z}^p$ but for the same space-time 
dimension $\nu_0$, and we comment on relations between integrals in different dimensions 
at the end of this section.

\subsection{Total derivatives in dimensional regularisation}
We start by presenting a theorem first used in refs.~\cite{Tkachov:1981wb,Chetyrkin:1981qh} to study linear relations among 
integrals.  In order to state the theorem, it is useful to introduce the notation
\beq
I(p_1,\ldots,p_E;m_1^2,\ldots,m_p^2;\nu;D) = \int\left(\prod_{j=1}^Le^{\gamma_E\eps}
\frac{\rd^Dk_j}{i\pi^{D/2}}\right)\,F_I(k_1,\ldots,k_L;p_1,\ldots,p_E;\nu;D)\,,
\eeq
with
\beq
\,F_I(k_1,\ldots,k_L;p_1,\ldots,p_E;\nu;D) = 
\frac{\cN(\{k_j\cdot k_l, k_j\cdot p_l\};D)}{\prod_{j=1}^p(m_j^2-q_j^2-i\varepsilon)^{\nu_j}}\,.
\eeq

\begin{prop}\label{prop:ibp}
In dimensional regularisation, we have 
\beq\label{eq:IBPgen}
\int {\rd}^Dk_i\,\frac{\partial}{\partial k_i^\mu}
\left[v^{\mu}\,F_I(k_1,\ldots,k_L;p_1,\ldots,p_E;\nu;D)\right]=0\,,\qquad 1\le i\le L\,,
\eeq
for every $D$-dimensional vector $v^{\mu}$.
\end{prop}
\begin{proof}We follow the argument given in ref.~\cite{Lee:2008tj} and use the simplified notation
\beq
I = \int \rd^Dk_i\,F_I(k_i)\,.
\eeq
We expect the integral to be invariant under general linear changes of variables, 
$k_i^{\mu}\to \lambda\,k_i^{\mu}+v^{\mu}$, with $\lambda$ a non-zero real number and 
$v^{\mu}$ a $D$-dimensional vector independent of $k^{\mu}$:
\beq
I = \lambda^D\,\int \rd^Dk_i\,F_I(\lambda\,k_i+v)\,.
\eeq
\begin{itemize}
\item Under an infinitesimal translation, $k_i^{\mu}\to k_i^{\mu}+\varepsilon v^{\mu}$, we have
\beq
I = \int \rd^Dk_i\,F_I(k_i+\varepsilon v) = I + \varepsilon v^{\mu}\int \rd^Dk_i\,
\frac{\partial}{\partial k_i^{\mu}}F_I(k_i) + \mathcal{O}(\varepsilon^2)  \,.
\eeq
Hence, we must have:
\beq
\int \rd^Dk_i\,\frac{\partial}{\partial k_i^{\mu}}\left[v^{\mu}F_I(k_i)\right]=0\,, 
\textrm{~~with $v^{\mu}$ independent of $k_i^{\mu}$}\,.
\eeq
\item Under an infinitesimal rescaling, $k_i^{\mu}\to e^{\varepsilon}k_i^{\mu}$, we have
\beq
I = e^{D\varepsilon}\int \rd^Dk_i\,F_I(e^{\varepsilon}k_i) = I + 
\varepsilon\left[D\,I + \int \rd^Dk_i\,k_i^{\mu}\frac{\partial}{\partial k_i^{\mu}}F_I(k_i) \right] 
+ \mathcal{O}(\varepsilon^2)\,.
\eeq
Hence, we must have
\beq
\int \rd^Dk_i\,\frac{\partial}{\partial k_i^{\mu}}\left[k_i^{\mu}F_I(k_i)\right]=0\,.
\eeq
\end{itemize}
\end{proof}
\begin{rem} We have stated Proposition~\ref{prop:ibp} for a specific loop momentum $k_i$, but it clearly also holds 
for any of the other loop momenta. 
Note that the regularisation scheme plays an important role, and Proposition~\ref{prop:ibp} may 
be false in another scheme. For example, if we use a cut-off as a regulator, 
${\rd}^Dk_i\to {\rd}^Dk_i\,\theta(\Lambda^2-|k_i^2|)$, the integral is not invariant under shifts 
$k_i^{\mu} \to k_i^{\mu} + v^{\mu}$, and Proposition~\ref{prop:ibp} does not necessarily hold
as it might need to be corrected by boundary terms.
\end{rem}

\subsection{Integration-by-parts relations}\label{sec:ibprels}
We now show that the vanishing of the surface terms in Proposition~\ref{prop:ibp} implies that 
Feynman integrals satisfy linear recursion relations in the propagator exponents $\nu$. 
We first note that, when the differential operator $\frac{\partial}{\partial k_i^{\mu}}v^{\mu}$ 
acts on the propagator $(m_j^2-q_j^2)^{-\nu_j}$, it shifts the value of 
the exponent. For example, if $v\neq k_i$ and $q_j = k_i+p_1$, we  have:
\begin{equation}
\frac{\partial}{\partial k_i^{\mu}}\left[\frac{v^{\mu}}{(m_j^2-q_j^2)^{\nu_j}}\right] 
= 2\nu_j\frac{(k_i+p_1)\cdot v}{(m_j^2-q_j^2)^{\nu_j+1}}\,.
\end{equation}
We see that we recover the same propagator, with the power raised by one unit. 
However, we have also introduced a numerator, which may not be present in our original Feynman integral.
If $v$ is chosen to be a loop or external momentum (or any linear combination thereof), 
then the numerator is a polynomial in the dot products involving loop and/or external momenta. 
We may then express these dot products as inverse propagators. For example, if we choose $v=p_1$ 
in the above example, we have
\begin{align}\begin{split}\label{eq:derIBP}
\frac{\partial}{\partial k_i^{\mu}}&\left[\frac{v^{\mu}}{(m_j^2-q_j^2)^{\nu_j}}\right] 
= 2\nu_j\frac{k_i\cdot p_1+p_1^2}{(m_j^2-q_j^2)^{\nu_j+1}}= \nu_j\frac{q_j^2-k_i^2+p_1^2}{(m_j^2-q_j^2)^{\nu_j+1}}\\
&= \nu_j\Big[-\frac{1}{(m_j^2-q_j^2)^{\nu_j}}+\frac{1}{(m_1^2-q_1^2)^{-1}(q_j^2-m_j^2)^{\nu_j+1}} + \frac{p_1^2-m_1^2+m_j^2}{(q_j^2-m_j^2)^{\nu_j+1}}\Big]\,.
\end{split}\end{align}
We see that we have indeed obtained a linear combination of propagators raised to different powers 
(we assume without loss of generality that there is a propagator with momentum $q_1=k_i$ and mass $m_1$). 
In general, the coefficients of the linear combination are polynomial functions in the external scales 
(and the dimensional regulator $\eps$).
Numerator factors involving loop momenta appear as propagators 
raised to negative integer powers. 
However, we had to \emph{assume} that we can write all scalar products in 
terms of inverse propagators. 
This may not always be the case, and we say that a family of Feynman integrals 
is \emph{complete} if it contains enough propagators to express all dot products involving at least 
one loop  momentum in terms of inverse 
propagators (we also use the word \emph{topology} to refer to a complete family of Feynman integrals, 
while elsewhere in the literature it is sometimes used for families that may not be complete).
Note that this is the same assumption we made in \cref{sec:paramRep} when discussing the Cutkosky-Baikov representation
and, as we argued there, every family of Feynman integrals can be 
completed by adding enough propagators. We therefore always assume from now on that our families of Feynman 
integrals are complete. This discussion can then be summarised as follows:
\begin{prop}
Every complete family of Feynman integrals (that is, every topology) satisfies linear recursion relations in the propagator exponents 
$\nu\in\mathbb{Z}^p$, called \emph{integration-by-parts (IBP)} relations~\cite{Tkachov:1981wb,Chetyrkin:1981qh}. 
The coefficients of the linear combinations are rational functions in the external scales $x$ 
and the dimensional regulator $\epsilon$.
\end{prop}
IBP relations allow one to relate different Feynman integrals from the same family to each other, and therefore to reduce considerably 
the number of Feynman integrals to be computed. 
These relations involve only rational functions of the scales and~$\epsilon$.
It is possible to solve the recursion relations and to 
express every member of a given (complete) family in terms of a basis of integrals. 
Elements of such a basis are conventionally called \emph{master integrals}. 
\begin{prop}
The number of master integrals is always finite.
\end{prop}
The proof of this result can be found in refs.~\cite{Smirnov:2010hn,Bitoun:2017nre}. 
It is also possible to predict the dimension of the basis (up to some caveats) via the number of 
critical points~\cite{Lee:2013hzt} or a certain Euler characteristic~\cite{Bitoun:2017nre}.
We stress nevertheless that there is considerable freedom in how one can choose a basis of integrals 
for the complete family (the dimension of the basis, however, is fixed). 
In fact, different choices may lead to improved methods for 
their evaluation, as will be discussed in \cref{sec:diffeqs}.

To close this subsection, let us introduce some concepts that will be useful later on. We define a map
\beq
\vartheta: \mathbb{Z}^p \to \{0,1\}^p\,; \quad\nu \mapsto \vartheta(\nu) 
= \big(\theta(\nu_1),\ldots,\theta(\nu_p)\big)\,,
\eeq
where $\theta(x)$ denotes the Heaviside step function:
\beq
\theta(x) =\left\{\begin{array}{ll}
1\,,&  \textrm{ if $x>0$}\,,\\
0\,, & \textrm{ else}\,.
\end{array}\right.
\eeq
We say that two Feynman integrals $I(x;\nu;D)$ and $I(x;\nu';D)$ \emph{belong to the same sector}
if $\vartheta(\nu)=\vartheta(\nu')$. Roughly speaking, a \emph{sector} is the collection of all Feynman
integrals of a family that share the same set of active propagators (where we qualify a propagator as 
\emph{active} if it is raised to a strictly positive power $\nu_i$, i.e., $\theta(\nu_i)=1$). Integrals from the 
same sector may, however, differ by the choice of the numerator factors. There is a natural partial 
order on sectors, and we say that $\vartheta(\nu)\ge \vartheta(\nu')$ if $\nu_i\ge \nu'_i$, for all $1\le i\le p$. 
If we derive IBP relations starting from Proposition~\ref{prop:ibp} for a given integral $I(x;\nu;D)$, 
then these relations will only involve integrals from the sector $\vartheta(\nu)$ or from lower sectors. 
It can also happen that a Feynman integral can be expressed as a linear combination of integrals from 
lower sectors only. We call such an integral \emph{reducible}. If all integrals from a given sector 
are reducible, we call the sector reducible.\footnote{\label{fn:sectors}
While we are of course free to choose
reducible integrals as basis elements, we will from now on assume that this is not done.
In particular, we will always choose master integrals that cannot be written as a linear combination
of integrals with fewer propagators, that is no master integral belongs to a reducible sector.} 

\begin{rem}\label{rem:neededIBPs}
In applications it is often not necessary to solve the IBP relations in all sectors. 
Indeed, one often encounters the situation that certain propagators only enter with negative 
exponents (e.g., because they were added to obtain a complete family of integrals). 
In such a scenario it is then only necessary to solve the IBP relations in the subsectors 
that describe the active propagators needed for the application one has in mind (with a caveat discussed in section~\ref{sec:other_rels}).
\end{rem}

\begin{rem}\label{rem:cutIBPs}
	IBP relations for cut Feynman integrals (see Remark~\ref{rem:cutFeynmanIntegrals})
	follow from those for uncut integrals. In an IBP relation,
	the `cut' only acts on Feynman integrals and not on the rational
	functions. According to Proposition~\ref{prop:cuts},
	some of the Feynman integrals in the IBP relation might be
	set to 0. In particular, this implies that the maximal cut of a 
	reducible integral always vanishes.
\end{rem}

\begin{rem}[Lorentz-invariance identities]\label{rem:LII} At the heart of the existence of linear relations among 
Feynman integrals is the fact that first-order differential operators act via shifting the propagator exponents. 
This is of course not restricted to derivatives with respect to loop momenta, as considered in 
Proposition~\ref{prop:ibp}, but it is easy to see that the same conclusion holds for derivatives with respect 
to external momenta (or even propagator masses). In particular, this means that the fact that the generator of a
Lorentz transformation annihilates a Feynman integral (cf.~eq.~\eqref{eq:LI}) translates also into a linear 
relation among Feynman integrals with shifted exponents. These so-called \emph{Lorentz-invariance identities} 
are not independent from the IBP relations generated by Proposition~\ref{prop:ibp}~\cite{Lee:2008tj}. 
\end{rem}


\subsection{Solving IBP relations}\label{sec:solIBPs}

In all but the simplest cases, it is not known how to find a solution to 
the IBP relations in closed form. 
That is, we do not know how to find an expression for an integral
$I(x;\nu;D)$ for generic $\nu$ in terms of a basis of master integrals.
In practical applications, however, we are only interested
in solving the relations for a finite set of $\nu$
(e.g., the ones appearing in the calculation of an amplitude).
As noted in ref.~\cite{Laporta:2000dsw}, instead of solving a recursion problem,
this can be formulated as a linear algebra problem:
one explicitly writes down all IBP relations satisfying some criterion
(typically one puts a bound on the sum of the negative and positive
exponents $\nu_i$ ; we note that care must be taken when choosing this bound,
as if it is too constraining one might miss some relations and find a number of
master integrals that is larger than it should be), which includes all the 
integrals one wishes to rewrite,
and then solves for the `complicated' integrals in terms of the `simple'
integrals. The notion of `complicated' and `simple' integrals requires introducing
an ordering criterion, and it is natural to favour integrals with 
fewer propagators (see~\cref{fn:sectors}). 
The approach to solving the IBP relations proposed in~ref.~\cite{Laporta:2000dsw},
commonly called the \emph{Laporta algorithm},
has since been implemented in a number of public programs such as \texttt{AIR}~\cite{Anastasiou:2004vj},
\texttt{REDUZE2}~\cite{vonManteuffel:2012np},
\texttt{LiteRed}~\cite{Lee:2013mka},
\texttt{FIRE}~\cite{Smirnov:2019qkx}, 
and \texttt{Kira}~\cite{Klappert:2020nbg}
(we only refer to the latest releases of each code).

For integrals with several loops and many scales, solving the
IBP relations is still a very challenging problem and often poses a
bottleneck in their evaluation. The main difficulty comes from the fact
that the number of IBP relations to solve grows very fast, and 
the complexity of the rational functions involved can become unmanageable
when the number of scales increases. Many approaches
have been proposed to overcome these challenges, starting
with the implementation of sophisticated algorithms to solve
such systems of equations in the programs mentioned above.
Aside from these algorithms, there have also been new developments
which are more specifically targeted to IBP relations.
First, one can try to construct simpler IBP relations by choosing 
vectors $v^\mu$ in~\cref{eq:IBPgen} with certain properties. 
For instance, one might choose $v^\mu$ so that
the IBP relations do not involve integrals with propagators raised
to higher powers (this would be natural for reducing amplitudes
to a basis of master integrals, since in that case one is in general trying to
reduce numerator factors)~\cite{Gluza:2010ws}. It is not trivial
to find such vectors $v^\mu$, but it is a question that can be answered
with tools from computational algebraic geometry,
namely by solving syzygy equations~\cite{Schabinger:2011dz,Ita:2015tya,Larsen:2015ped,Georgoudis:2016wff},
which have had a strong impact
in modern approaches to the calculation of scattering amplitudes
\cite{Abreu:2017xsl,Abreu:2017hqn,Abreu:2021asb,Badger:2021imn}.
Second, one might try to choose a basis for which the coefficients
in the IBP relations will be simpler. It has been observed that working
with a so-called \emph{canonical basis} (see Proposition~\ref{prop:pure} below) can be very beneficial
as the singularity structure of the coefficients greatly simplifies~\cite{Usovitsch:2020jrk,Boehm:2020ijp,Smirnov:2020quc},
and in particular the dependence on $D$ and on $x$ completely factorises.
Combined with multivariate partial-fractioning techniques~\cite{Pak:2011xt,Meyer:2016slj,Abreu:2019odu,Boehm:2020ijp,Heller:2021qkz}, one
obtains much more compact solutions to the IBP relations.
Finally, and closely connected to the point above about the simplicity of the coefficients, one can solve
the IBP relations for numerical values of $x$ and $D$, and use
this numerical data to determine the analytic expressions 
\cite{vonManteuffel:2014ixa,Peraro:2016wsq,Peraro:2019svx,Klappert:2020aqs,DeLaurentis:2022otd}. This
approach is particularly powerful when combined with the second
point mentioned above: using numerical evaluations
bypasses the intermediate analytic complexity generated in the 
solution of the IBP system, and one directly reconstructs the much
simpler rational expressions appearing in the solution of the
IBP relations. The numerical evaluations are most commonly
done in finite-field arithmetic~\cite{vonManteuffel:2014ixa,Peraro:2016wsq}, 
for which there exist very efficient 
linear algebra algorithms and where all numerical calculations
are exact, thus removing any questions about numerical precision and
stability.

Let us conclude by mentioning that an alternative approach is being developed to achieve the reduction to master integrals without the need to solve the IBP relations. We will review this new approach in section~\ref{sec:intersection}. Currently, however, this new approach is still in its infancy, and solving the IBP relations using the techniques described in this section remains the method of choice for all applications.


\subsection{Example: The one-loop bubble integral}\label{sec:oneLoopBubIBP}

As an example of the application of IBP relations, we will consider the one-loop bubble integral:
\begin{equation}\label{eq:twoMassBubInt}
	I(p^2;m_1^2,m_p^2;\nu_1,\nu_2;D) = 
	e^{\gamma_E\eps}\int\frac{\rd^Dk}{i\pi^{D/2}}
	\frac{1}{(m_1^2-k^2-i\varepsilon)^{\nu_1}(m_2^2-(k+p)^2-i\varepsilon)^{\nu_2}}\,.
\end{equation}
Note that any polynomial $\cN(\{k^2, p_1\cdot k\};D)$ can always be written as a polynomial
in the propagators and in the variables in $x=(p^2,m_1^2,m_2^2)$, so without loss of generality
we consider the integral with a trivial numerator.

The IBP relations can be constructed from
\begin{equation}
	\int \rd^Dk\frac{\partial}{\partial k^{\mu}}
	\left[v^\mu\frac{1}{(m_1^2-k^2-i\varepsilon)^{\nu_1}(m_2^2-(k+p)^2-i\varepsilon)^{\nu_2}}\right]=0\,,
\end{equation}
with $v^\mu=k^\mu$ and $v^\mu=p^\mu$. For $v^\mu=k^\mu$, we obtain
\begin{align}\begin{split}\label{eq:ibpbub1}
	&(D-2\nu_1-\nu_2)\,I(\nu_1,\nu_2)-\nu_2\,I(\nu_1-1,\nu_2+1)-\nu_2(p^2-m_1^2-m_2^2)\,I(\nu_1,\nu_2+1)\\
	&\,+2\nu_1m_1^2\,I(\nu_1+1,\nu_2)=0\,,
\end{split}\end{align}
and for $v^\mu=p^\mu$ we obtain
\begin{align}\begin{split}\label{eq:ibpbub2}
	&(\nu_1-\nu_2)\,I(\nu_1,\nu_2)-\nu_1\,I(\nu_1+1,\nu_2-1)-\nu_1(p^2+m_1^2-m_2^2)\,I(\nu_1+1,\nu_2)\\
	&+\nu_2\,I(\nu_1-1,\nu_2+1)-\nu_2(m_1^2-p^2-m_2^2)\,I(\nu_1,\nu_2+1)=0\,,
\end{split}\end{align}
where we simplified the notation to only keep the dependence of the integrals on the $\nu_i$, i.e., 
we set $I(\nu_1,\nu_2)=I(p^2;m_1^2,m_p^2;\nu_1,\nu_2;D)$.

To simplify the discussion, let us first consider the case where $m_1^2=m_2^2=0$. 
Before exploring the IBP relations, we note that
in this limit $I(\nu_1,\nu_2)=0$ if either $\nu_1\leq0$ or $\nu_2\leq0$,
as under those conditions the integrals become scaleless.
The IBP relations above can be rewritten as:
\begin{align}\begin{split}\label{eq:ibpbubmassless}
	&I(\nu_1,\nu_2)=-\frac{\nu_1+\nu_2-1-D}{p^2(\nu_2-1)}\,I(\nu_1,\nu_2-1)-\frac{1}{p^2}\,I(\nu_1-1,\nu_2)\,,
	\qquad \textrm{for }\nu_2\neq1\,,\\
	&I(\nu_1,\nu_2)=-\frac{\nu_1+\nu_2-1-D}{p^2(\nu_1-1)}\,I(\nu_1-1,\nu_2)-\frac{1}{p^2}\,I(\nu_1,\nu_2-1)\,,
	\qquad \textrm{for }\nu_1\neq1\,.
\end{split}\end{align}
This means that any integral $I(\nu_1,\nu_2)$ is either 0 or can be
related to $I(1,1)$, which is trivial to compute (e.g., using Feynman parameters):
\begin{equation}
	I(p^2;0,0;1,1;D)=e^{\gamma_E\eps}(-p^2)^{\frac{D-4}{2}}\frac{\Gamma(2-D/2)\,\Gamma(D/2-1)^2}{\Gamma(D-2)} \,.
\end{equation}
We then find that this family of Feynman integrals contains a single master integral.
We also note that for this particular case it is not hard to compute the Feynman integral
for arbitrary powers of the propagators:
\begin{align}\begin{split}\label{eq:masslessBubGen}
	I(p^2;0,0;\nu_1,\nu_2;D)=&\,e^{\gamma_E\eps}(-p^2)^{-\nu_1-\nu_2+D/2}\\
	&\times\frac{\Gamma(\nu_1+\nu_2-D/2)}{\Gamma(D-\nu_1-\nu_2)}
	\frac{\Gamma(D/2-\nu_1)}{\Gamma(\nu_1)}\frac{\Gamma(D/2-\nu_2)}{\Gamma(\nu_2)} \,.
\end{split}\end{align}
From this representation, we recognize the IBP relations in eq.~\eqref{eq:ibpbubmassless} as consequences 
of the well-known recurrence relations between $\Gamma$-functions, $\Gamma(1+z)=z\,\Gamma(z)$.

If instead $m_1^2\neq0$ and $m_2^2=0$, then $I(\nu_1,\nu_2)=0$ if $\nu_1\leq0$.
The IBP relations for this case are obtained from \cref{eq:ibpbub1,eq:ibpbub2}
evaluated at $m_2^2=0$, and we find that there are two master integrals, which we can
choose to be $I(1,0)$ and $I(1,1)$, to which any $I(\nu_1,\nu_2)$ can be related.
The tadpole integral $I(1,0)$ is
\begin{equation}\label{eq:tad}
	I(m^2;1;D)=e^{\gamma_E\eps}\Gamma\left(1-D/2\right)(m^2)^{-1+D/2}\,,
\end{equation}
while the bubble integral $I(1,1)$ is
\begin{equation}\label{eq:oneMBubAllOrder}
	I(p^2;m_1^2,0;1,1;D)=e^{\gamma_E\eps}(m_1^2)^{-2+D/2}\frac{\Gamma(2-D/2)}{D/2-1}
	\,_2F_1\left(1,2-\frac{D}{2};\frac{D}{2};\frac{p^2}{m_1^2}\right)\,,
\end{equation}
where $_2F_1$ is Gauss' hypergeometric function:
\begin{equation}\label{eq:2f1def}
\,_2F_1\left(\alpha,
	\beta;\gamma;x\right)
	=
	\frac{\Gamma(\gamma)}{\Gamma(\alpha)\Gamma(\gamma-\alpha)}
	\int_0^1 u^{\alpha-1} 
	(1-u)^{\gamma-\alpha-1}
	(1-xu)^{-\beta} \rd u\,.
\end{equation}
Just as for the bubble with massless propagators, the IBP relations for this 
case follow from the recurrence relations (known in this case as contiguous relations) satisfied by the 
$_2F_1$ hypergeometric function~\cite{handbook}.

In the case where $m_1^2,m_2^2\neq 0$, $I(\nu_1,\nu_2)=0$ only if both 
$\nu_1\leq0$ and $\nu_2\leq0$. Using the IBP relations in \cref{eq:ibpbub1,eq:ibpbub2}, 
we would find that we can  relate any integral of the form $I(\nu_1,\nu_2)$ to three 
master integrals, which can be chosen to be $I(1,0)$, $I(0,1)$ and $I(1,1)$.
The expression for the tadpoles $I(1,0)$ and $I(0,1)$ is given in \cref{eq:tad},
and the bubble integral can be written as:
\begin{align}\label{eq:twoMBubAllOrder}
\nonumber	I(p^2;m_1^2,m_2^2;1,1;D)=
	&\,e^{\gamma_E\eps}\frac{\Gamma(2-D/2)}{D/2-1}\frac{(-p^2)^{-2+D/2}}{w-\bar{w}}\\
	&\,\left[(-w\bar{w})^{D/2-1}\,_2F_1\left(1,-2+D;\frac{D}{2};\frac{w}{w-\bar{w}}\right)\right.\\
\nonumber	&\left.- \left(-(1-w)(1-\bar{w})\right)^{D/2-1}\,_2F_1\left(1,-2+D;\frac{D}{2};\frac{w-1}{w-\bar{w}}\right)\right]\,,
\end{align}
where $w$ and $\bar{w}$ can be implicitly defined through
\begin{equation}\label{eq:var2mbub}
	(1-w)(1-\bar{w})=\frac{m_1^2}{p^2}\,,\qquad
	w\bar{w}=\frac{m_2^2}{p^2}\,,\qquad
	w-\bar{w}=\sqrt{\lambda\left(1,\frac{m_1^2}{p^2},\frac{m_2^2}{p^2}\right)}\,,
\end{equation}
with $\lambda(a,b,c)=a^2+b^2+c^2-2ab-2ac-2bc$.


\subsection{Dimension-shift relations}

So far we have only considered linear relations connecting Feynman integrals with different propagator
exponents $\nu$, but with the same values for the space-time dimension~$D$ and external scales~$x$. 
The Feynman-parameter representation of Feynman integrals in \cref{eq:Feynman}, however, makes it manifest that there is no 
substantial difference between the dimension~$D$ and the exponents $\nu_i$: they both appear as 
exponents of the polynomials in the integrand. More precisely, the space-time dimension enters the 
exponents of the Feynman parameter integral only through the combination $\nu_0 = D_0/2$. It is therefore 
natural to expect that there are linear relations relating Feynman integrals with different values of $\nu_0$. 
This was worked out in detail in refs.~\cite{Tarasov:1996br,Lee:2009dh}.

\begin{prop}\label{prop:dim_shift}
For Feynman integrals depending on generic and non-zero propagator masses, we have:
\beq\label{eq:dimUp}
I(x;\nu;D-2) = 
(-1)^{L}\,\cU\left(\frac{\partial}{\partial m_1^2},\ldots,\frac{\partial}{\partial m_p^2}\right)\, I(x;\nu;D).
\eeq
\end{prop}
\noindent
Before we give a derivation of this relation, let us make some comments. 
The differential operator in the right-hand side involves the $\cU$-polynomial that appears in the 
Feynman and Schwinger parametrisations, but with the Feynman/Schwinger parameters $\alpha_i$ replaced 
by the differential operators $\frac{\partial}{\partial m_i^2}$. For this reason we need to consider integrals 
with generic propagator masses. The action of this differential operator produces on the right-hand side 
a linear combination of Feynman integrals in $D$ dimensions with shifted propagator exponents. 
Proposition~\ref{prop:dim_shift} asserts that this linear combination equals the Feynman integral in 
$(D-2)$ dimensions (up to an overall factor). Moreover, once all derivatives have been carried out, 
we can set the masses to non-generic (possibly zero) values, and we obtain a relation between 
integrals in different dimensions also for non-generic masses.

\begin{proof}
We start from the Schwinger parametrisation in eq.~\eqref{eq:Schwinger}. 
The dependence of the integrand in eq.~\eqref{eq:Schwinger} on the space-time dimension and the 
masses is particularly simple: the space-time dimension only enters through the exponent of the 
$\cU$-polynomial, and the masses only appear in the exponent in the integrand, see in particular
\cref{eq:FmassDep}.
We then have
\beq
\cU\left(\frac{\partial}{\partial m_1^2},\ldots,\frac{\partial}{\partial m_p^2}\right)
e^{-\cF(\alpha;x)/\cU(\alpha)} = 
(-1)^{L}\,\cU\left(\alpha_1,\ldots,\alpha_p\right)e^{-\cF(\alpha;x)/\cU(\alpha)}\,.
\eeq
We see that the application of the differential operator amounts to multiplication by $\cU(\alpha)$ 
in the integrand, changing the exponent from $-\nu_0$ to $-(\nu_0-1)$.
\end{proof}

\begin{prop}\label{prop:dim_shift_down}
An integral in $D+2$ dimensions can be written as a linear combination of integrals in $D$
dimensions as:
\beq\label{eq:dimDown}
I(x;\nu;D+2) = \frac{2^L\,G(p_1,\ldots,p_{E-1})}{(D-K+1)_L}\mathcal{B}(b_1,\ldots,b_K)I(x;\nu;D)\,,
\eeq
where $(x)_L$ is the Pochhammer symbol,
$\mathcal{B}$ is the Baikov polynomial defined in \cref{eq:baikovPol}, 
and the $b_i$ are operators that lower
the value of the exponent $\nu_i$, that is 
\begin{equation*}
b_i^a\,I(x;\nu;D)=I(x;\nu_1,\ldots,\nu_i-a,\ldots;D)\,.
\end{equation*}
\end{prop}
\begin{proof}
We start from the Cutkosky-Baikov parametrisation in \cref{eq:Baikov}, and note that the
space-time dimension appears in a very simple way. In particular, in the integrand it only
appears in the exponent of the Baikov polynomial. The dimension-shifting relation
of \cref{eq:dimDown} then  follows simply by noting that a polynomial on the inverse-propagator
variables in the numerator gives a linear combination of integrals with
shifted powers of the propagators.
\end{proof}
\noindent

\begin{ex}[Dimension-shift relations for the one-loop bubble]
Let us return to the example of \cref{sec:oneLoopBubIBP}, the one-loop
bubble integral $I(p^2;m_1^2,m_2^2;1,1;D)$. Using \cref{eq:dimUp}, we can
write the $(D-2)$-dimensional integral as a linear combination of integrals in $D$ dimensions.
Then, using the IBP relations of \cref{eq:ibpbub1,eq:ibpbub2}, the latter can be rewritten 
in terms of a set of master integrals. Using the basis of \cref{sec:oneLoopBubIBP}, we get
\begin{align}\begin{split}\label{eq:dimShiftBubUp}
	I(p^2;m_1^2,m_2^2;1,1;D-2)&=I(2,1;D)+I(1,2;D)\\
	&=\frac{1}{\lambda(p^2,m_1^2,m_2^2)}\bigg[ -\frac{m_1^2+p^2-m_2^2}{2m_1^2}(D-2)I(1,0;D)\\
	&-\frac{m_2^2+p^2-m_1^2}{2m_2^2}(D-2)I(0,1;D)+2\,p^2(D-3)\,I(1,1;D)\bigg]\,,
\end{split}\end{align}
where we used the same compact notation as in \cref{eq:ibpbub1,eq:ibpbub2}, but this time
keeping the dependence on $D$ explicit.
The factor of $\lambda(p^2,m_1^2,m_2^2)$ is in this case a singularity related to the integral on the left-hand side; however, IBP relations can generically introduce spurious poles as well. 
We can also relate the $(D+2)$-dimensional bubble to the master integrals in $D$ dimensions
with \cref{eq:dimDown}. Using the same basis as above, we find:
\begin{align}
	I(p^2;m_1^2,m_2^2;1,1;D+2)&=\frac{1}{2p^2(D-1)}\bigg[I(-1,1;D)+I(1,-1;D)\nonumber\\
	&+\lambda(p^2,m_1^2,m_2^2)I(1,1;D)+(p^2+m_1^2-m_2^2)I(1,0;D)\nonumber\\
	&+(p^2+m_2^2+m_1^2)I(0,1;D)
	\bigg]\nonumber\\
	&=-\frac{m_1^2+p^2-m_2^2}{2p^2(D-1)}I(1,0;D)
	-\frac{m_2^2+p^2-m_1^2}{2p^2(D-1)}I(0,1;D)\nonumber\\
	&+\frac{\lambda(p^2,m_1^2,m_2^2)}{2p^2(D-1)}\,I(1,1;D)\,.
\end{align}
The dimension-shift relations for the cases where $m_1^2=0$ and/or $m_2^2=0$ are obtained
from the above by simply setting the masses and the scaleless integrals to zero. As for
the IBP relations, the dimension-shift relations correspond to recurrence relations
of the special functions the integrals evaluate to, either gamma functions or (generalised) hypergeometric
functions (see \cref{eq:masslessBubGen} for an explicit example). 
These recurrence relations are now with respect to the parameter $\nu_0=D/2$,
but, as already pointed out, there is no substantial difference between the exponents $\nu$
and $\nu_0$.
\end{ex}

\subsection{Other relations among Feynman integrals}
\label{sec:other_rels}

Let us conclude this section with some comments about to what extent the IBP and dimensional-shift 
relations capture all relations between Feynman integrals. 
Since IBP and dimensional-shift relations are linear, we need to discuss linear and non-linear relations separately.

Let us start by discussing linear relations among Feynman integrals.  
Currently, there is no indication of homogeneous linear relations among Feynman integrals in the same family, in dimensional regularisation,
that do not follow from the IBP and dimensional-shift relations, and conjecturally we have
\begin{conj}
All linear relations among Feynman integrals from a given complete family follow from IBP 
and dimension-shift relations.
\end{conj}
As an example, we already pointed out in Remark~\ref{rem:LII} that the Lorentz-invariance 
identities follow from the IBP relations. In applications, however, there may be various caveats:
\begin{itemize}
\item We have already mentioned that it is useful to separate Feynman integrals into sectors, 
and in concrete applications one only needs to consider those sectors where a certain subset 
of propagators is active. 
Consequently, one only needs to solve the IBP identities in those sectors. It can happen, however, 
that certain relations among integrals from lower sectors are only found if IBP relations involving 
integrals in higher sectors are considered. Hence, if IBP relations from higher sectors are neglected, 
it may appear that there are relations among Feynman integrals in lower sectors that seem not to follow 
from IBP relations (while in fact they do follow from IBP relations in higher sectors).
\item It can also happen that certain new relations arise in limits where the external scales take degenerate values. 
Those additional relations arise from IBP relations once the degeneracy among the scales is resolved, 
cf.,~e.g.,~ref.~\cite{Kalmykov:2011yy,Kniehl:2012hn}. Similar types of relations were derived in ref.~\cite{Kniehl:2016yrh} from functional equations, although the examples shown there relate integrals of different families and so fall outside the scope of the conjecture above.
\item IBP relations detect linear relations of Feynman integrals in generic space-time dimensions $D$. 
In applications one is usually interested in external momenta that lie in 4 space-time dimensions. 
Since at most 4 vectors can be linearly independent in 4 dimensions, it may happen that certain combinations of 
scales vanish when the external momenta are chosen four-dimensional. 
This may lead to new relations for four-dimensional external momenta, 
which were not present for  $D$-dimensional external momenta. We note that, due to momentum conservation, 
this situation arises for the first time for 6 external momenta in 4 space-time dimensions.
\end{itemize}

Finally, let us comment on non-linear relations among Feynman integrals. 
Much less is known about such relations, though over the last couple of years several 
instances of nontrivial quadratic relations 
between Feynman integrals 
(or their maximal cuts) have been discovered~\cite{Broadhurst:2016hbq,Broadhurst:2018tey,Zhou:2017jnm,Zhou:2020glw,Lee:2018jsw,Bonisch:2021yfw}. 
By nontrivial, we mean both that it is not possible to obtain these 
relations as a consequence of linear relations, and that some of the loop integrations do not trivially factorise.
As we will mention in section~\ref{sec:intersection}, the existence of nontrivial quadratic relations seems to 
be very general, and it would be interesting to explore them more generally in the future. Currently, 
there is no example of nontrivial relations of higher degree (cubic or higher).


\section{The method of differential equations}
\label{sec:diffeqs}

The IBP relations reviewed in the previous section allow one to reduce 
the problem of computing a given family of Feynman integrals to the computation 
of a (finite) set of basis integrals, usually called master integrals. 
The master integrals must be evaluated by other means, and
there are various well-established methods to compute them 
(cf.,~e.g., ref.~\cite{Smirnov:2012gma} and references therein for a review).
Some of these methods rely on direct integration of parametric representations
such as the ones in \cref{eq:Schwinger,eq:Feynman}, using various methods to perform the integrals 
(see,~e.g.,~refs.~\cite{Brown:2008um,Anastasiou:2013srw,Panzer:2014gra,Panzer:2014caa,
Bogner:2014mha,Bogner:2015nda,Ablinger:2014yaa,Hidding:2017jkk,Broedel:2019hyg}), 
but over the last decades the method of differential 
equations~\cite{Kotikov:1990kg,Kotikov:1991hm,Kotikov:1991pm,Gehrmann:1999as,Henn:2013pwa,Papadopoulos:2014lla} 
has established itself as one of the most powerful. 
There are several good reviews and lecture series on how to use differential equations
to compute Feynman integrals, see,~e.g.,~refs.~\cite{Henn:2014qga,Henn:2014yza}. 
Here, we attempt to give an overview of the general strategy in a broader mathematical context.

\subsection{Differential equations satisfied by Feynman integrals}

In the following we denote by 
$\vec{\mathcal{I}}(x,\eps)=\big(\mathcal{I}_1(x,\eps),\ldots,\mathcal{I}_N(x,\eps)\big)^T$ 
a vector of basis integrals depending on the scales $x=(x_1,\ldots,x_s)$, 
and we assume that the entries of $\vec{\mathcal{I}}(x,\eps)$ are ordered in a way 
that is compatible with the natural order on the sectors.

Let us consider the derivative of $\vec{\mathcal{I}}(x,\eps)$ with respect to an external 
scale $x_i$.
We can exchange the derivative $ \partial_{x_i}= \frac{\partial}{\partial x_i} $ and the loop integration, and we act with the 
derivative on the loop \emph{integrand}. 
If $x_i$ is a propagator mass, $x_i=m_j^2$, the action on the integrand is easy to compute. 
If $x_i$ is a scalar product between external momenta, $x_i=p_j\cdot p_k$, 
then we can use the chain rule to express $\partial_{x_i}$ in terms of the differential 
operators $p_j^\mu\frac{\partial}{\partial p_k^\mu}$, whose action on the loop integrand
is straightforward and very similar to the calculation done in \cref{eq:derIBP},
see also Remark~\ref{rem:LII}.
The expression one obtains after acting with the differential operators can then be
rewritten in terms of master integrals using the IBP identities described in \cref{sec:ibps}.
In summary, the derivative of master integrals with respect to an external scale $x_i$ 
can be expressed as a linear combination of master integrals. That is, we can write:
\beq
\partial_{x_i}\vec{\mathcal{I}}(x,\eps) = A_{x_i}(x,\eps)\,\vec{\mathcal{I}}(x,\eps)\,,
\eeq 
where the $A_{x_i}(x,\eps)$ are $N\times N$ matrices.
Since IBP relations involve only rational coefficients, the entries of $A_{x_i}(x,\eps)$ 
are rational functions in $x$ and $\eps$. 
Moreover, if the entries of $\vec{\mathcal{I}}(x,\eps)$ are ordered such that they respect 
the natural partial order on the sectors (as we are assuming), then the matrices are block upper-triangular.

\begin{ex}[The differential equations for the one-loop bubble]\label{ex:diffEq1lbNPure}
Let us return once more to the example of the one-loop bubble integral with two massive propagators.
We already established in \cref{sec:oneLoopBubIBP} that there are three master integrals associated 
with this topology. We set
\begin{equation}
    \vec{\mathcal{I}}(p^2,m_1^2,m_2^2;D)=\big({I}(1,1),{I}(1,0),{I}(0,1)\big)^T\,,
\end{equation}
where we use the propagator powers in \cref{eq:twoMassBubInt} to distinguish the three different master integrals. 
From now on we will drop the dependence on all other quantities. 
The derivatives with respect to $p^2$, $m_1^2$ and $m_2^2$ are given by
\begin{align}\begin{split}\label{eq:dep2}
    \partial_{p^2}\,\vec{\cI}=&\,
    \begin{pmatrix}
    \frac{I(0,2)}{2p^2}-\frac{I(1,1)}{2p^2}+\frac{(p^2-m_1^2+m_2^2)I(1,2)}{2p^2}\\
    0\\
    \frac{I(-1,2)}{2p^2}-\frac{I(0,1)}{2p^2}+\frac{(p^2-m_1^2+m_2^2)I(0,2)}{2p^2}
    \end{pmatrix}\\
    =&\,\begin{pmatrix}
    \frac{(D-4)p^4+2(m_1^2+m_2^2)p^2-(D-2)(m_1^2-m_2^2)^2}{2p^2\lambda(p^2,m_1^2,m_2^2)}&
    \frac{(D-2)(p^2+m_2^2-m_1^2)}{2p^2\lambda(p^2,m_1^2,m_2^2)}&
    \frac{(D-2)(p^2+m_1^2-m_2^2)}{2p^2\lambda(p^2,m_1^2,m_2^2)}\\
    0&0&0\\
    0&0&0\\
    \end{pmatrix}\vec{\cI}\,,
\end{split}\end{align}
\begin{align}\begin{split}\label{eq:dem12}
    \partial_{m_1^2}\,\vec{\cI}=&\,
    \begin{pmatrix}
    -I(2,1)\\
    -I(2,0)\\
    0
    \end{pmatrix}
    =\begin{pmatrix}
    \frac{(D-3)(m_1^2-p^2-m_2^2)}{\lambda(p^2,m_1^2,m_2^2)}&
    \frac{(D-2)(m_1^2+m_2^2-p^2)}{2m_1^2\lambda(p^2,m_1^2,m_2^2)}&
    -\frac{D-2}{2p^2\lambda(p^2,m_1^2,m_2^2)}\\
    0&\frac{D-2}{2m_1^2}&0\\
    0&0&0\\
    \end{pmatrix}\vec{\cI}\,,
\end{split}\end{align}
\begin{align}\begin{split}\label{eq:dem22}
    \partial_{m_2^2}\,\vec{\cI}=&\,
    \begin{pmatrix}
    -I(1,2)\\
    0\\
    -I(0,2)
    \end{pmatrix}
    =\begin{pmatrix}
    \frac{(D-3)(m_2^2-p^2-m_1^2)}{\lambda(p^2,m_1^2,m_2^2)}&
    -\frac{D-2}{2p^2\lambda(p^2,m_1^2,m_2^2)}&
    \frac{(D-2)(m_1^2+m_2^2-p^2)}{2m_1^2\lambda(p^2,m_1^2,m_2^2)}&\\
    0&0&0\\
    0&0&\frac{D-2}{2m_2^2}\\
    \end{pmatrix}\vec{\cI}\,,
\end{split}\end{align}
where $\lambda(a,b,c)$ was defined below \cref{eq:var2mbub}. In the equations above,
we first present the action of the differential operator on the integrand of \cref{eq:twoMassBubInt},
and then what one obtains after using IBP relations to rewrite the expressions in terms of master integrals.
Note in particular that $\partial_{p^2}\cI(0,1)$ is only explicitly zero after accounting for IBP relations.
\end{ex}

It is convenient to package the different partial derivatives into a total differential 
with respect to all external scales, $\rd = \sum_{i=1}^s\rd x_i\partial_{x_i}$:
\beq\label{eq:DEQ}
\rd\,\vec{\mathcal{I}}(x,\eps) = A(x,\eps)\,\vec{\mathcal{I}}(x,\eps)\,,
\eeq 
where $A(x,\eps) = \sum_{i=1}^s\rd x_i\, A_{x_i}(x,\eps)$ is a matrix whose entries are rational one-forms.

\begin{rem}The matrices $A_{x_i}$ are in fact not independent. 
Indeed, the total differential must satisfy $\rd^2=0$, and we have
\beq
0=\rd^2 \,\vec{\mathcal{I}}(x,\eps) = \rd\!\left[ A(x,\eps)\,\vec{\mathcal{I}}(x,\eps)\right] 
= \left[\rd A(x,\eps) - A(x,\eps)\wedge A(x,\eps)\right]\,\vec{\mathcal{I}}(x,\eps)\,.
\eeq
It follows that $A(x,\eps)$ must satisfy the \emph{integrability condition}
\beq\label{eq:DEQ_integrability}
\rd A(x,\eps) - A(x,\eps)\wedge A(x,\eps) = 0\,,
\eeq
where `$\wedge$' denotes the wedge product between differential forms. 
Equation~\eqref{eq:DEQ_integrability} gives a set of differential relations between the matrices 
$A_{x_i}(x,\eps)$ that can serve as a useful check of the correctness of the differential equations. 
It is straightforward to verify, for instance, 
that the matrices in \cref{eq:dep2,eq:dem12,eq:dem22} satisfy these relations.
\end{rem}

\begin{rem}
The differential operators $\partial_{x_i}$ are not all independent, 
even if the scales $x_i$ are. Indeed, for every integral $I(x,\nu,\eps)$ 
(basis integral or not) we have the relation:
\beq\label{eq:scaling_relation}
\sum_{i=1}^sx_i\,\partial_{x_i}I(x,\nu,\eps) = \frac{\alpha_I}{2}\,I(x,\nu,\eps) \,,
\eeq 
where $\alpha_I = \big[I(x,\nu,\eps)\big]$ was defined in eq.~\eqref{eq:Rescaling}. 
To see why eq.~\eqref{eq:scaling_relation} holds, we note that the differential 
operator on the left-hand side is the infinitesimal generator of the dilatations 
$x\mapsto \lambda^2\,x$, 
\beq
\sum_{i=1}^sx_i\,\partial_{x_i} = \frac{1}{2}\,\mathcal{D}\,,
\eeq
where $\mathcal{D}$ was defined in eq.~\eqref{eq:dilatations_generator}.
The eigenvalues of the dilatation operator are the 
mass dimensions (the additional factor of $1/2$ comes from the fact that $\alpha_I$ 
is the mass dimension of the integral, and the scales $x_i$ themselves have 
mass dimension equal to~2).

Equation~\eqref{eq:scaling_relation} implies that the nontrivial functional dependence 
of $I(x,\nu,\eps)$ can only be in the ratios $y_i = x_i/x_s$, $1\le i\le s-1$. 
Indeed, if we change variables from 
$(x_1,\ldots,x_s)$ to $(y_1,\ldots,y_s) = (x_1/x_s,\ldots,x_{s-1}/x_s,x_s)$, 
we can easily check that we have the relation:
\beq
\sum_{i=1}^sx_i\,\partial_{x_i} = y_s\,\partial_{y_s} \,,
\eeq 
so that eq.~\eqref{eq:scaling_relation} implies
\beq
I(x,\nu,\eps) = y_s^{\alpha_I/2}\,\hat{I}(y_1,\ldots,y_{s-1},\nu,\eps) 
= x_s^{\alpha_I/2}\,\hat{I}(x_1/x_s,\ldots,x_{s-1}/x_s,\nu,\eps) \,.
\eeq
It is therefore sufficient to consider the derivatives with respect to the ratios 
$y_i$, $1\le i\le s-1$, rather than in the individual scales $x_i$ 
(or, equivalently, we may put $x_s=1$ during the computation). 
We note that this implies that the differential equation satisfied by a one-scale 
integral ($s=1$) is trivial, and the method of differential equations is not suitable 
for computing such integrals. 
These integrals must then either be computed by other means, e.g.,  direct integration or the \emph{dimensional recurrence and analyticity method}~\cite{Lee:2009dh,Lee:2012te}.
If one still wishes to use differential equations,
it is often possible to introduce an additional scale $x_2$. 
The method of differential equations may then be applied to the ratio $y_1=x_1/x_2$, 
and we recover the original integral in the limit $y_1\to \infty$. 
We refer to ref.~\cite{Henn:2013nsa} for details.

Equation~\eqref{eq:scaling_relation} leads to the following differential relation between the master integrals
\beq\label{eq:scaling_relation_matrix}
\sum_{i=1}^sx_i\,\partial_{x_i}\,\vec{\mathcal{I}}(x,\eps) = \frac{1}{2}\big[\vec{\mathcal{I}}(x,\eps)\big]\,\vec{\mathcal{I}}(x,\eps)\,,
\eeq
where $\big[\vec{\mathcal{I}}(x,\eps)\big] = \textrm{diag}\big(\big[{\mathcal{I}}_1(x,\eps)\big],\ldots, \big[{\mathcal{I}}_N(x,\eps)\big]\big)$ is the diagonal matrix whose entries are the mass dimensions of the basis elements. 
Equation \eqref{eq:scaling_relation_matrix} provides another nontrivial check
of the correctness of the differential equations.
For instance, using the matrices in \cref{eq:dep2,eq:dem12,eq:dem22}, we verify
that
\begin{equation}
    \left(p^2\,\partial_{p^2}+m_1^2\,\partial_{m_1^2}+m_2^2\,\partial_{m_2^2}\right)\,\vec{\cI}
    =\frac{1}{2}\begin{pmatrix}
    D-4&0&0\\
    0&D-2 &0\\
    0&0&D-2
    \end{pmatrix}\,\vec{\cI}\,,
\end{equation}
in the case where $\vec{\cI}$ denotes the master integrals for the two-mass bubble
discussed in \cref{sec:oneLoopBubIBP}.
\end{rem}

\subsubsection{Change of basis}
The differential equation~\eqref{eq:DEQ} can be very hard to solve in closed form, 
including the full dependence on the dimensional regulator $\eps$. 
In applications we are only interested in the Laurent expansion of $\vec{\mathcal{I}}(x,\eps)$ 
up to some finite order in $\eps$, and this order is typically relatively low. 
The basis $\vec{\mathcal{I}}(x,\eps)$ is not unique, and we may use this freedom to change basis 
to a new basis $\vec{\mathcal{J}}(x,\eps)$ which brings the differential equation into a form 
that can be more easily solved. 
We are therefore interested in determining how the differential 
equation~\eqref{eq:DEQ} behaves under a change of basis. 

Since $\vec{\mathcal{I}}(x,\eps)$ and $\vec{\mathcal{J}}(x,\eps)$ are bases for the same vector space, there must be an invertible matrix $M(x,\eps)$ such that 
\beq\label{eq:change_of_basis}
\vec{\mathcal{I}}(x,\eps) = M(x,\eps)\,\vec{\mathcal{J}}(x,\eps)\,.
\eeq
The new basis $\vec{\mathcal{J}}(x,\eps)$ satisfies the differential equation
\beq
\rd\vec{\mathcal{J}}(x,\eps) = A'(x,\eps)\,\vec{\mathcal{J}}(x,\eps)\,,
\eeq 
where the matrix $ A'(x,\eps)$ is related to the matrix $ A(x,\eps)$ from eq.~\eqref{eq:DEQ} by
\beq\label{eq:gauge_transformation}
 A'(x,\eps) = M(x,\eps)^{-1}\left[ A(x,\eps) M(x,\eps)-\rd  M(x,\eps)\right]\,.
\eeq
At this point we have to make an important comment: we have already mentioned that, 
due to the rational nature of the IBP relations,  the original matrix $A(x,\eps)$
in eq.~\eqref{eq:DEQ} is a matrix of \emph{rational} one-forms. 
However, we have not specified what the functional dependence of the 
matrix $M(x,\eps)$ is, and no-one forces us to choose $M(x,\eps)$ to have rational entries. 
On the other hand, if the entries of $M(x,\eps)$ are not rational, the new basis 
$\vec{\mathcal{J}}(x,\eps)$ will not consist of Feynman integrals of the form $I(x,\nu,\eps)$
multiplied by rational functions of $x$ and $\epsilon$. 
In the following, we will refer to a basis related to one consisting of integrals of the form
$I(x,\nu,\eps)$ via a rational transformation matrix $M(x,\eps)$ as an \emph{IBP-basis}.
We will only consider changes of bases $M(x,\eps)$ that depend rationally on $\eps$, 
and we will call the transformation in eq.~\eqref{eq:change_of_basis} \emph{rational} 
(or \emph{algebraic}, or \emph{transcendental}) if the entries of $M(x,\eps)$ are {rational}
(or {algebraic}, or {transcendental}) functions of $x$.\footnote{In the case of a transcendental 
transformation, we actually allow the entries of $M(x,\eps)$ to be also {rational} or {algebraic}, i.e., 
there is a strict inclusive hierarchy among rational, algebraic and transcendental transformations.}

\begin{ex}[Change of basis for the one-loop bubble]
Let us return to the example of the one-loop bubble with two massive propagators,
for which we chose the basis (see \cref{sec:oneLoopBubIBP} and \cref{ex:diffEq1lbNPure})
\begin{equation}
    \vec{\mathcal{I}}(p^2,m_1^2,m_2^2;D)=\big({I}(1,1),{I}(1,0),{I}(0,1)\big)^T\,.
\end{equation}
We will later see that it is particularly convenient to consider the change of basis
\begin{equation}\label{eq:cob2mb}
    \vec{\mathcal{I}}(p^2,m_1^2,m_2^2;D)=\frac{2}{2-D}
    \begin{pmatrix}
    \dfrac{1}{\sqrt{\lambda(p^2,m_1^2,m_2^2)}}&0&0\\
    0&1&0\\
    0&0&1
    \end{pmatrix}
    \vec{\mathcal{J}}(p^2,m_1^2,m_2^2;D)\,,
\end{equation}
with $\lambda(a,b,c)$ as defined below \cref{eq:var2mbub}.
This is an example of an algebraic transformation.

For the case where one of the propagators is massless, say $m_2^2=0$, 
the basis is two-dimensional. In \cref{sec:oneLoopBubIBP}
we chose the basis
\begin{equation}
    \vec{\mathcal{I}}(p^2,m_1^2;D)=\big({I}(1,1),{I}(1,0)\big)^T\,.
\end{equation}
For this example, the change of basis equivalent to that of \cref{eq:cob2mb} is
\begin{equation}
    \label{eq:cob1mb}
    \vec{\mathcal{I}}(p^2,m_1^2;D)=\frac{2}{2-D}
    \begin{pmatrix}
    \dfrac{1}{p^2-m_1^2}&0\\
    0&1
    \end{pmatrix}
    \vec{\mathcal{J}}(p^2,m_1^2;D)\,.
\end{equation}
This is an example of a rational transformation.

In both cases above, we can compute the matrices $A'(x,D)$ associated with the bases $\vec{\mathcal{J}}(x;D)$ using
\cref{eq:gauge_transformation}, and we find that
\begin{equation}\label{eq:canBub}
    A'(x,D) = \frac{2-D}{2}\,\widetilde{A}(x)\,.
\end{equation}
That is, the $D$-dependence completely factorises from $A'(x,D)$, and for $D=2-2\epsilon$
the matrix is proportional to $\epsilon$. As we will later see, such bases
have particularly interesting properties.

\end{ex}

\subsection{Solving the differential equations}
\label{sec:solve_deq}

In this section we discuss a strategy for solving the differential equations satisfied by the master integrals (for an alternative strategy in the case of one-variables problems, see,~e.g.,~ref.~\cite{Ablinger:2018zwz}). 
We focus here on analytic approaches, though we point out that it is also possible to solve the differential 
equations numerically, cf.,~e.g.,~refs.~\cite{Aglietti:2007as,Mistlberger:2018etf,Mandal:2018cdj,
Moriello:2019yhu,Hidding:2020ytt,Liu:2020kpc,Liu:2021wks,Liu:2022chg,Liu:2022mfb}.

We are interested in the first few terms in the Laurent series around $\eps=0$ 
of ${\vec{\mathcal{I}}}(x,\eps)$. Typically, some basis integrals will have poles at $\eps=0$, 
but it is generally easy to determine the first non-zero order in the Laurent expansion. 
More generally, we have:
\begin{prop} \label{prop:eps-regular}
For every IBP-basis, there is a rational transformation to an IBP-basis whose elements are finite and non-zero at $\eps=0$.
\end{prop}
\noindent
A basis with the property that every basis element is finite and non-zero at $\epsilon=0$ is called 
$\eps$-\emph{regular}~\cite{Lee:2019wwn} (see also ref.~\cite{Chetyrkin:2006dh,vonManteuffel:2014qoa} for the 
closely related notions of $\eps$-\emph{finite} and \emph{quasi-finite} basis).
The proof of this statement is constructive and provides an algorithm to determine the 
$\eps$-\emph{regular} basis~\cite{Lee:2019wwn}. Note that, since the transformation obtained from this algorithm is rational, 
the $\eps$-regular basis will also be an IBP-basis. In the following we therefore assume without 
loss of generality that the starting IBP-basis ${\vec{\mathcal{I}}}(x,\eps)$ is $\eps$-regular. 
This implies that the matrix $A(x,\eps)$ in the differential equation~\eqref{eq:DEQ} 
is finite for $\eps=0$, and we define
\beq
A(x,\eps) = A_0(x) + A_1(x,\eps)\,,\textrm{~~~with~~~} \lim_{\eps\to0}A_1(x,\eps) = 0\,.
\eeq

Let us now return to the discussion of how to solve he differential equation.
We recall that we assumed that the ordering among the basis elements in $\vec{\mathcal{I}}(x,\eps)$ 
is compatible with the natural order on the sectors. We can then write
\beq
\vec{\mathcal{I}}(x,\eps) = \big(\vec{\mathcal{I}}_{\theta_1}(x,\eps), \ldots, \vec{\mathcal{I}}_{\theta_{N_{\textrm{sec}}}}(x,\eps)\big)^T\,.
\eeq
The $\theta_i\in\{0,1\}^p$ label the $N_{\textrm{sec}}$ irreducible sectors, and 
the irreducible basis integrals in the sector $\theta_i$ are collected in
$\vec{\mathcal{I}}_{\theta_i}(x,\eps) = \big(I(x,\nu^{i}_1,\eps), \ldots, I(x,\nu^{i}_{N_{\theta_i}},\eps)\big)^T$,
where  $\nu^{i}_j=(\nu^{i}_{j,1},\ldots,\nu^{i}_{j,p})$.
The irreducible basis integrals in $\vec{\mathcal{I}}_{\theta_i}(x,\eps)$ have all the  propagators of the sector ($\vartheta(\nu^i_j) =\theta_i$), 
and no linear combination of them can be reduced to a lower sector via IBP relations. 
Since the matrix $A(x,\eps)$ is block upper-triangular, we can split the homogeneous system of 
differential equations in eq.~\eqref{eq:DEQ} into $N_{\textrm{sec}}$ inhomogeneous systems 
for the basis in a given sector:
\beq\label{eq:DEQ_sector}
\rd \vec{\mathcal{I}}_{\theta_i}(x,\eps) = A_{\theta_i}(x,\eps) \vec{\mathcal{I}}_{\theta_i}(x,\eps) + \vec{\mathcal{N}}_{\theta_i}(x,\eps)
\,,
\eeq
where 
$\vec{\mathcal{N}}_{\theta_i}(x,\eps)$ collects  contributions from lower sectors. Written in this form, the equations can be solved sector by sector. The procedure involves two steps, which we describe separately.

\paragraph{The associated homogeneous equation.} 
We start by considering the homogenous equation for $\eps=0$ associated to eq.~\eqref{eq:DEQ_sector}:
\beq\label{eq:DEQ_sector_homogeneous}
\rd \vec{\mathcal{I}}_{\theta_i,h}(x) = A_{\theta_i,0}(x) \vec{\mathcal{I}}_{\theta_i,h}(x) \,, 
\quad \text{where }\,A_{\theta_i,0}(x) = \lim_{\eps\to0}A_{\theta_i}(x,\eps)\,.
\eeq
Since this is a linear system of $N_{\theta_i}$ homogeneous first-order differential equations, the general solution takes the form:
\beq\label{eq:general_sector_solution}
\vec{\mathcal{I}}_{\theta_i,h}(x) = \sum_{k=1}^{N_{\theta_i}}c_k\,\vec{\mathcal{I}}_{\theta_i,h}^{(k)}(x)\,,
\eeq
where the $c_k$ are unknown complex constants and the $\vec{\mathcal{I}}_{\theta_i,h}^{(k)}(x)$ 
form a basis for the solution space of eq.~\eqref{eq:DEQ_sector_homogeneous}. 
We can then form a matrix where the columns are these vectors, called the \emph{Wronskian matrix} of 
eq.~\eqref{eq:DEQ_sector_homogeneous} (sometimes also called the \emph{fundamental solution matrix}):
\beq\label{eq:whomogeneous}
W_{\theta_i}(x) = \left(\vec{\mathcal{I}}_{\theta_i,h}^{(1)}(x), \ldots, \vec{\mathcal{I}}_{\theta_i,h}^{(N_{\theta_i})}(x)\right)\,.
\eeq
The general solution in eq.~\eqref{eq:general_sector_solution} can then be cast in the form
\beq
\vec{\mathcal{I}}_{\theta_i,h}(x) = W_{\theta_i}(x)\,\vec{c}\,,
\eeq
where $\vec{c} = (c_1,\ldots, c_{N_{\theta_i}})^T$ is determined by the initial condition, 
and the Wronskian matrix satisfies the equation
\beq
\rd W_{\theta_i}(x) = A_{\theta_i,0}(x)\,W_{\theta_i}(x)\,.
\eeq
We see that the entire information about the general solution to the homogeneous 
equation~\eqref{eq:DEQ_sector_homogeneous} is encoded into the Wronskian matrix, 
which can be interpreted as a matrix solution to the homogeneous equation at $\eps=0$. 
Since the columns of $W_{\theta_i}(x)$ form a basis for the solution space, 
the determinant of $W_{\theta_i}(x)$ must be non-zero (at least for generic values of $x$). 
The determinant satisfies the differential equation
\beq
\rd\det W_{\theta_i}(x) = \left[\textrm{Tr} A_{\theta_i,0}(x) \right]\, \det W_{\theta_i}(x)\,.
\eeq
Although the determinant is often an algebraic function, the individual entries of $W_{\theta_i}(x)$, 
are in general not algebraic, but rather transcendental functions.

Determining the entries of the Wronskian matrix is often a very complicated task and no closed form 
for the solutions is known, especially if $N_{\theta_i}>2$. In that case it can be useful to turn the system 
of $N_{\theta_i}$ first-order equations for the vector $\vec{\mathcal{I}}_{\theta_i,h}(x)$ 
into a system of differential equations of higher order for an individual entry of 
$\vec{\mathcal{I}}_{\theta_i,h}(x)$. It may then be possible to obtain local power-series solutions 
from the Frobenius method using standard techniques, which may be continued to multi-valued solutions 
for all values of $x$. For recent applications in this direction in the context of Feynman integrals, 
see refs.~\cite{Bonisch:2021yfw,Vanhove:2018mto,Klemm:2019dbm,Bonisch:2020qmm}. 
In the following we will assume that we know the expression for $W_{\theta_i}(x)$, 
which is typically the case for $N_{\theta_i}\le 2$.

\begin{rem}
The entries of the Wronskian matrix $W_{\theta_i}(x)$ have an interpretation in terms of maximal cuts. 
Indeed,  the cuts of a Feynman integral satisfy the same differential equations as the full, uncut, integral, 
where we have to put to zero all cut integrals where we would need to cut propagators raised to 
non-positive powers (cf.,~e.g.,~refs.~\cite{Anastasiou:2002yz,Anastasiou:2003yy}). 
It then follows that if we cut all propagators in the sector $\theta_i$, the resulting maximal cuts of 
$\vec{\mathcal{I}}_{\theta_i}(x,\eps)$, denoted $\vec{\mathcal{I}}_{\theta_i}^{\textrm{max-cut}}(x,\eps)$, 
satisfy the homogeneous equation associated to
eq.~\eqref{eq:DEQ_sector}~\cite{Harley:2017qut,Bosma:2017ens,Primo:2016ebd,Primo:2017ipr,Frellesvig:2017aai}:
\beq\label{eq:max_cut_deq}
\rd \vec{\mathcal{I}}_{\theta_i}^{\textrm{max-cut}}(x,\eps) = A_{\theta_i}(x,\eps)\, 
\vec{\mathcal{I}}_{\theta_i}^{\textrm{max-cut}}(x,\eps) \,,
\eeq
Note that eq.~\eqref{eq:max_cut_deq} holds in $D=D_0-2\eps$ dimensions, not just for $\eps=0$.

The number of independent maximal cuts in the sector $\theta_i$ always agrees with the number 
$N_{\theta_i}$ of irreducible basis elements in that sector. We can then interpret the 
Wronskian matrix $W_{\theta_i}(x)$ as the matrix of all independent maximal cuts at $\eps=0$ 
in the sector $\theta_i$. This implies that we can also obtain the entries of the $W_{\theta_i}(x)$ 
by evaluating maximal cuts~\cite{Harley:2017qut,Bosma:2017ens,Primo:2016ebd,Primo:2017ipr,Frellesvig:2017aai}.
It is in fact possible to say a bit more about the homogeneous equation in eq.~\eqref{eq:max_cut_deq}, 
at least at a conjectural level. In ref.~\cite{Lee:2018jsw} it was conjectured that it is always 
possible to find an \emph{algebraic} transformation $M(x,\eps)$ such that $A_{\theta_i}(x,\eps)$ has the form
\beq
A_{\theta_i}(x,\eps) = (\eps+\mu_{\theta_i})\,S_{\theta_i}(x)\,, \quad \mu_{\theta_i} = 0\textrm{ or }1/2 
\textrm{ and } S_{\theta_i}(x)^T = S_{\theta_i}(x)\,,
\eeq
for every irreducible sector $\theta_i$.
\end{rem}

\paragraph{Iterative solution of the inhomogeneous equation.} 
We have established that we can associate the Wronskian matrix $W_{\theta_i}(x)$ 
to equation~\eqref{eq:DEQ_sector}. 
Let us now consider the change of basis 
$\vec{\mathcal{I}}_{\theta_i}(x,\eps) = W_{\theta_i}(x)\vec{\mathcal{J}}_{\theta_i}(x,\eps)$
(we recall that $W_{\theta_i}(x)$ is invertible, so this is a valid change of basis).
We obtain:
\beq\label{eq:DEQ_sector_J}
\rd \vec{\mathcal{J}}_{\theta_i}(x,\eps) = A_{\theta_i,1}'(x,\eps) 
\vec{\mathcal{J}}_{\theta_i}(x,\eps) + \vec{\mathcal{N}}'_{\theta_i}(x,\eps)\,,
\eeq
with 
\beq\begin{split}
A_{\theta_i,1}'(x,\eps) &\,= W_{\theta_i}(x)^{-1} \left[A_{\theta_i}(x,\eps)-A_{\theta_i,0}(x)\right]W_{\theta_i}(x)\,,\\
\vec{\mathcal{N}}'_{\theta_i}(x,\eps) &\,= W_{\theta_i}(x)^{-1}\,\vec{\mathcal{N}}_{\theta_i}(x,\eps)\,.
\end{split}\eeq
Note that, since the entries of $W_{\theta_i}(x)$ will generally be transcendental functions, 
this change of basis will in general be a transcendental transformation, and the new basis will no longer be an IBP-basis.
Since our basis is $\eps$-regular, all quantities admit a Taylor series expansion close to $\eps=0$:
\beq\begin{split}
\vec{\mathcal{J}}_{\theta_i}(x,\eps) &\,= \sum_{k=0}^{\infty}\eps^k\,\vec{\mathcal{J}}_{\theta_i}^{(k)}(x)\,,\\
\vec{\mathcal{N}}_{\theta_i}'(x,\eps) &\,= \sum_{k=0}^{\infty}\eps^k\,\vec{\mathcal{N}}_{\theta_i}^{'(k)}(x)\,,\\
A_{\theta_i,1}'(x,\eps) &\,= \sum_{k=1}^{\infty}\eps^k\,A_{\theta_i,1}^{\prime(k)}(x)\,.
\end{split}\eeq
Equation~\eqref{eq:DEQ_sector_J} can now easily be solved order by order in $\eps$. At $\mathcal{O}(\eps^0)$ we find:
\beq
\rd \vec{\mathcal{J}}_{\theta_i}^{(0)}(x) = \vec{\mathcal{N}}^{\prime(0)}_{\theta_i}(x)\,.
\eeq
This equation is solved by quadrature:
\beq\label{eq:J0_sol}
\vec{\mathcal{J}}_{\theta_i}^{(0)}(x) = \vec{\mathcal{J}}_{\theta_i,x_0}^{(0)} + 
\int_{x_0}^x\vec{\mathcal{N}}^{\prime(0)}_{\theta_i}(x_1)\,,
\eeq
where $\vec{\mathcal{J}}_{\theta_i,x_0}^{(0)}$ is a constant vector of initial conditions 
and the integration runs over a (fixed) path from an initial point $x_0$ to a generic point $x$. Note that the integrability condition in eq.~\eqref{eq:DEQ_integrability} ensures that the solution only depends on the choice of the end points $x_0$ and $x$, but it does not depend on the details of the path.
We will discuss how to fix the initial conditions in 
section~\ref{eq:solving_deqs}.\footnote{The initial condition $\vec{\mathcal{J}}_{\theta_i,x_0}^{(0)}$ 
can be thought of as the value of $\vec{\mathcal{J}}_{\theta_i}^{(0)}(x)$ at $x=x_0$. 
It is often useful to choose as initial point $x_0$ a point where 
$\vec{\mathcal{J}}_{\theta_i}^{(0)}(x)$ is singular, in which case the situation requires special attention, 
because in that case the value at $x=x_0$, as well as the integral in eq.~\eqref{eq:J0_sol}, are ill-defined. 
We will discuss this in more detail in section~\ref{sec:it_regularisation}. 
For now, it suffices to assume that $\vec{\mathcal{J}}_{\theta_i}^{(0)}(x)$ is regular at $x=x_0$.} 
At $\mathcal{O}(\eps^1)$ we find the equation:
\beq
\rd \vec{\mathcal{J}}_{\theta_i}^{(1)}(x) = \vec{\mathcal{N}}^{\prime(1)}_{\theta_i}(x) + 
A_{\theta_i,1}^{\prime(1)}(x)\,\left[\vec{\mathcal{J}}_{\theta_i,x_0}^{(0)}
+\int_{x_0}^x\vec{\mathcal{N}}^{\prime(0)}_{\theta_i}(x_1)\right]\,.
\eeq
The solution is
\beq
\label{eq:J1_sol}
\vec{\mathcal{J}}_{\theta_i}^{(1)}(x) = \vec{\mathcal{J}}_{\theta_i,x_0}^{(1)} 
+ \int_{x_0}^x\left[\vec{\mathcal{N}}^{\prime(1)}_{\theta_i}(x_1)+A_{\theta_i,1}^{\prime(1)}(x_1)\, 
\vec{\mathcal{J}}_{\theta_i,x_0}^{(0)}\right]+\int_{x_0}^xA_{\theta_i,1}^{\prime(1)}(x_1)
\int_{x_0}^{x_1}\vec{\mathcal{N}}^{\prime(1)}_{\theta_i}(x_2)\,.
\eeq
We can of course continue in this way and write down the solution at every order in~$\eps$. 
We see that at every order we have one extra integration compared to the previous order, 
which naturally leads us to a class  of functions called \emph{iterated integrals}. 
They can be defined as follows: consider a path $\gamma: [t_0,t] \to X$ in some space $X$, 
and $\omega_1,\ldots \omega_n$ differential one-forms on $X$. 
We define their iterated integral along $\gamma$ by~\cite{ChenSymbol}
\beq\label{eq:iterated_def}
    \int_{\gamma}\omega_1\cdots \omega_n := \int_{t_0}^t \rd t_n\,f_n(t_n)
    \int_{t_0}^{t_{n}}\rd t_{n-1}\,f_{n-1}(t_{n-1})\int_{t_0}^{t_{n-1}}\cdots 
    \int_{t_0}^{t_2}\rd t_1\,f_1(t_1)\,.
\eeq
where $f_i:\mathbb{C}\to \mathbb{C}$ are complex functions which are the pull-backs 
of the $\omega_i$ along $\gamma$, $\gamma^*\omega_i = \rd t_i\,f_i(t_i)$. 
We will study iterated integrals and their properties in detail in section~\ref{sec:iterated}. 
Our arguments lead to the following conclusion:
\begin{prop}
\label{prop:iterated}
At every order in $\eps$, the basis integrals $\vec{\mathcal{I}}(x,\eps)$ involve rational 
functions in $x$, the maximal cuts for $\eps=0$, and iterated integrals of these functions.
\end{prop}


\subsection{Differential equations in canonical form}
\label{sec:can_deq}
The strategy spelled out in the previous section is very general, 
but it is often not the most efficient approach to find an explicit expression for the 
Feynman integrals.
It has been suggested in ref.~\cite{Henn:2013pwa} 
that there is a class of distinguished bases in which the differential equations take a particularly 
simple form. The existence of this form for the differential equation is still conjectural, 
though the conjecture is supported by all existing computations of multi-loop Feynman integrals. 
The idea is to find a transformation matrix $M(x,\eps)$ such that the matrix $A'(x,\eps)$ in 
eq.~\eqref{eq:gauge_transformation} is as simple as possible. The main conjecture is the following, 
first formulated in ref.~\cite{Henn:2013pwa} (in the slightly restricted setting discussed below):

\begin{conj}\label{conj:canonical}
For every IBP-basis $\vec \cI(x,\epsilon)$ satisfying the differential equation
\begin{equation}
    \rd \vec\cI(x,\epsilon)=A(x,\epsilon) \,\vec\cI(x,\epsilon)\,,
\end{equation} 
there is a (possibly transcendental) transformation to a new basis
$ \vec{\mathcal{I}}(x,\epsilon)=M(x,\eps)\,\vec{\mathcal{J}}(x,\epsilon)$ such that
\begin{equation}\label{eq:canonical_form}
     \rd \vec{\mathcal{J}}(x,\epsilon)=A'(x,\eps)\,\vec{\mathcal{J}}(x,\epsilon) = \eps\,\widetilde{A}(x)\,\vec{\mathcal{J}}(x,\epsilon)\,,
\end{equation} 
where $\widetilde{A}(x)$ is a matrix of one-forms with at most logarithmic singularities.
\end{conj}
\noindent
If the new basis satisfies eq.~\eqref{eq:canonical_form} we say that the differential equations 
are in \emph{canonical form}, and we call the basis a \emph{canonical basis}.
Here, a \emph{differential one-form with logarithmic singularities} should be thought of as a 
differential one-form $\omega_i$ such that, in some appropriate choice of local coordinates 
$\xi = (\xi_1,\ldots,\xi_s)$, all singularities take the form 
\beq\label{eq:log_singularity}
\omega_i \sim \dlog p(\xi) + \cdots = \sum_{j=1}^s\rd\xi_j\,\frac{\partial_{\xi_j}p(\xi)}{p(\xi)}+ \cdots\,,
\eeq
where $p$ is an algebraic function of $\xi$ and the dots denote
power-suppressed terms that are regular at $p(\xi)=0$.

 The simplest examples of canonical differential equations are then those with 
 \beq\label{eq:dlog_form}
 \widetilde{A}(x) = \sum_{i}A_{i}\,\dlog p_i(x)\,,
 \eeq
 where the $A_i$ are constant matrices and $p_i(x)$ are algebraic functions of the external scales. 
 We refer to this case as a \emph{canonical dlog-form}. 
 This is the setup originally discussed in ref.~\cite{Henn:2013pwa}. 
 It is known that there are Feynman integrals for which one cannot obtain a system of differential 
 equations in canonical dlog-form as in eq.~\eqref{eq:dlog_form}. 
 To see this, we start by arguing that, if there is a canonical dlog-form, 
 it can necessarily be reached by an \emph{algebraic} transformation. 
 Indeed, there must be a transformation $M(x,\eps)$ such that
 \beq
 \eps\,\widetilde{A}(x) = M(x,\eps)^{-1}\big[A(x,\eps)M(x,\eps) - \rd M(x,\eps)\big]\,.
 \eeq
We know that $A(x,\eps)$ is a matrix of rational one-forms. The matrix $\widetilde{A}(x)$ 
in eq.~\eqref{eq:dlog_form} is a matrix of algebraic one-forms, i.e., its entries are of 
the form $\sum_{k=1}^s \rd x_k\,a_k(x)$, with the $a_k(x)$ algebraic functions 
(simply take $a_k(x) = {\partial_{x}p(x)}/{p(x)}$). This shows that $M(x,\eps)$ 
must itself be algebraic (unless there are some really unnatural cancellations). 
Since we know that the solutions to the homogeneous differential equations for the maximal cuts 
(cf.~eq.~\eqref{eq:max_cut_deq}) generically involve transcendental functions already 
for $\eps=0$, it follows that the differential equations cannot always be cast in the 
form~\eqref{eq:dlog_form}. A necessary condition for a canonical $\dlog$-form to exist 
is thus that the maximal cuts for $\eps=0$ (which are the solutions to the 
homogeneous equations~\eqref{eq:max_cut_deq}) are all algebraic functions.

\begin{rem} It is an interesting question if the previous condition is also sufficient 
and not only necessary. For all known examples for which the homogeneous equations 
can be solved in terms of algebraic functions, a canonical $\dlog$-form also exists. 
From a purely mathematical standpoint, however, this is not automatic, because 
the matrix $\widetilde{A}(x)$ may involve one-forms without any singularities. 
It would be interesting to understand if, for Feynman integrals, the fact that all 
maximal cuts for $\eps=0$ are algebraic is equivalent to the existence of a canonical $\dlog$-form.
\end{rem}

\begin{ex}[Canonical basis for one-loop bubble]
    $\vec{\mathcal{J}}(p^2,m_1^2;D)$ and $\vec{\mathcal{J}}(p^2,m_1^2,m_2^2;D)$, as defined in
    \cref{eq:cob2mb,eq:cob1mb}, are canonical bases for $D=2-2\epsilon$ dimensions as noted
    in \cref{eq:canBub}. We will in fact later show that they are canonical $\dlog$ bases. Furthermore,
    the factor of $2/(2-D)$ ensures that the bases are regular (see Proposition~\ref{prop:eps-regular}).
    Canonical $\dlog$ bases for the one-loop bubble can also be obtained in $D=4-2\epsilon$
    by squaring propagators, as is clear from \cref{eq:dimShiftBubUp} which relates the
    bubble in $2-2\epsilon$ dimensions to the one in $4-2\epsilon$ dimensions.
    More generally, it is known how to find canonical $\dlog$ bases for any one-loop 
    integral \cite{Abreu:2017enx,Abreu:2017mtm,Arkani-Hamed:2017ahv,Caron-Huot:2021xqj,Chen:2022fyw}.
\end{ex}

\subsubsection{Solving canonical differential equations}
\label{eq:solving_deqs}
One of the main advantages when working with differential equations in canonical form
is that they are particularly easy to solve as a Laurent expansion in $\eps$. 
To see this, consider the differential equation
\beq\label{eq:canonical_deq}
\rd \vec{\mathcal{J}}(x,\eps) = \eps\,\widetilde{A}(x)\,\vec{\mathcal{J}}(x,\eps)\,,
\eeq
and assume that we know the value of $\vec{\mathcal{J}}_0(\eps) = \vec{\mathcal{J}}(x_0,\eps)$ 
at some point $x=x_0$. For simplicity, we assume that $\vec{\mathcal{J}}(x,\eps)$ is regular at 
$x=x_0$. We will return later to the situation where $\vec{\mathcal{J}}(x,\eps)$ develops 
singularities at this point. In order to get the value at some other point $x$, 
consider a path $\gamma$ from $x_0$ to $x$. The value of $\vec{\mathcal{J}}(x,\eps)$ 
is then obtained by parallel transporting the solution from $x_0$ to $x$ along $\gamma$:
\beq
\vec{\mathcal{J}}(x,\eps) = \mathbb{W}(\gamma,\epsilon)\vec{\mathcal{J}}_0(\eps)\,,
\eeq
where $\mathbb{W}(\gamma,\epsilon)$ is given by the path-ordered exponential:
\beq\label{eq:path_ordered_def}
\mathbb{W}(\gamma,\epsilon) = \mathbb{P}\exp\!\left[\eps\,\int_{\gamma}\widetilde{A}(x')\right]\,.
\eeq
Note that $\mathbb{W}(\gamma,\epsilon)$ is actually a Wronskian matrix for the differential equation (that is, a matrix built out of a basis of general
solutions to the differential equation like $W_{\theta_i}(x)$ 
in \cref{eq:whomogeneous} ; unlike $W_{\theta_i}(x)$, however,
$\mathbb{W}(\gamma,\epsilon)$ solves the full differential equation and not just
the homogeneous part). 
In particular, it satisfies eq.~\eqref{eq:canonical_deq}:
\beq
\rd \mathbb{W}(\gamma,\epsilon) = \eps\,\widetilde{A}(x)\,\mathbb{W}(\gamma,\epsilon)\,.
\eeq

It is straightforward to obtain the
 $\eps$-expansion of the path-ordered exponential. 
As an example, consider the case where $\widetilde{A}(x) = \widetilde{A}_1\,\omega_1+\widetilde{A}_2\,\omega_2$, 
with $\omega_i$ one-forms and $\widetilde{A}_i$ constant matrices. We have:
\begin{align}
\nonumber\mathbb{W}(\gamma,\epsilon)&\,= \mathbf{1} + \eps\left(\widetilde{A}_1\,
\int_{\gamma}\omega_1+\widetilde{A}_2\,\int_{\gamma}\omega_2\right)\\
\label{eq:W_exp}&\,+\eps^2\left(\widetilde{A}_1^2\,\int_{\gamma}\omega_1\omega_1+
\widetilde{A}_1\widetilde{A}_2\,\int_{\gamma}\omega_2\omega_1+\widetilde{A}_2\widetilde{A}_1\,
\int_{\gamma}\omega_1\omega_2+\widetilde{A}_2^2\,\int_{\gamma}\omega_2\omega_2\right)\\
\nonumber&\,+\eps^3\left(\widetilde{A}_1^3\,\int_{\gamma}\omega_1\omega_1\omega_1+
\widetilde{A}_1^2\widetilde{A}_2\,\int_{\gamma}\omega_2\omega_1\omega_1+
\widetilde{A}_1\widetilde{A}_2\widetilde{A}_1\,\int_{\gamma}\omega_1\omega_2\omega_1+
\widetilde{A}_2\widetilde{A}_1^2\,\int_{\gamma}\omega_1\omega_1\omega_2\right.\\
\nonumber&\,\left.+\widetilde{A}_2^2\widetilde{A}_1\,\int_{\gamma}\omega_1\omega_2\omega_2+
\widetilde{A}_2\widetilde{A}_1\widetilde{A}_2\,\int_{\gamma}\omega_2\omega_1\omega_2+
\widetilde{A}_1\widetilde{A}_2^2\,\int_{\gamma}\omega_2\omega_2\omega_1+
\widetilde{A}_2^3\,\int_{\gamma}\omega_2\omega_2\omega_2\right)\\
\nonumber&\,+\mathcal{O}(\eps^4)\,.
\end{align}
We see that at every order in $\eps$ the Laurent coefficients are iterated integrals, 
in agreement with the observation made in Proposition~\ref{prop:iterated}. 

We can actually say a bit more if we work with a canonical basis: since the matrix $\widetilde{A}(x)$ 
has at most logarithmic singularities, the iterated integrals arising from the path-ordered exponential 
$\mathbb{W}(\gamma,\epsilon)$ will diverge at most logarithmically, and in particular they have no poles or 
power-like singularities. More precisely, if the differential one-forms $\omega_i$ diverge at most 
like in eq.~\eqref{eq:log_singularity}, then their iterated integrals will behave in the limit $p(\xi)\to0$ like
\beq
\int_{\gamma}\omega_1\cdots\omega_n \sim \sum_{k=0}^nc_k\,\log^kp(\xi) + \cdots\,,
\eeq
where the $c_k$ are constants, and the dots indicate contributions that are power-suppressed in the limit. 
Iterated integrals of this type are often called \emph{pure functions} in the physics literature~\cite{ArkaniHamed:2010gh}. 
We thus conclude
\begin{prop}\label{prop:pure}
Every order in the $\eps$-expansion of a canonical basis involves iterated integrals that are pure functions. 
\end{prop}
Let us point out a major difference between Propositions~\ref{prop:iterated} and~\ref{prop:pure}. 
The fact that only pure functions arise in the $\eps$-expansion is tightly linked to a canonical basis, 
whose existence is still conjectural in the general case. Proposition~\ref{prop:iterated} is not conjectural, 
but it follows from direct considerations of how to solve the system of differential equations. 
By itself, however, Proposition~\ref{prop:iterated} does not allow us to conclude that the iterated integrals are pure.

\begin{rem} If a system of differential equations is in canonical form, then the integrability condition in eq.~\eqref{eq:DEQ_integrability} takes the form:
\beq
\eps\,\rd\widetilde{A}(x) -\eps^2\,\widetilde{A}(x)\wedge \widetilde{A}(x) = 0\,.
\eeq
Since this equation must hold for all values of $\eps$, the integrability condition takes a simpler form for systems in canonical from:
\beq\label{eq:DEQ_integrability_can}
\rd\widetilde{A}(x) = 0 \textrm{~~~~and~~~~} \widetilde{A}(x)\wedge \widetilde{A}(x) = 0\,.
\eeq
In particular, this shows that for a system of differential equations in canonical form, the matrix $\widetilde{A}(x)$ must be a matrix of \emph{closed} one-forms. This is manifestly the case for systems in canonical $\dlog$-form (because all $\dlog$-forms are closed). It was also observed to hold for non-$\dlog$ cases, see refs.~\cite{Adams:2018yfj,Broedel:2018rwm,Bogner:2014mha}.
\end{rem}

\begin{rem}
If the matrix $\widetilde{A}(x)$ satisfies the condition~\eqref{eq:DEQ_integrability_can},
the value of the path-ordered exponential does not depend on the choice of the path $\gamma$ (note that this
is in fact true for any basis: the integrability condition of \cref{eq:DEQ_integrability} is sufficient 
to guarantee path-independence of the solutions).  
More precisely, if $\gamma$ and $\gamma'$ are two paths that can be continuously deformed into 
one another without crossing any singularities, then 
$\mathbb{W}(\gamma,\epsilon) = \mathbb{W}(\gamma',\epsilon)$. 
We may then choose a path which makes the iterated integrals easy to evaluate in terms of known 
classes of special functions. In particular, we may write $\gamma$ as a composition of segments 
where all but one variable are constant. We refer to such a path as \emph{piecewise-constant}. 
On each segment the iterated integrals may be evaluated in terms of iterated integrals in one 
variable, which are easier to evaluate. The value of $\mathbb{W}(\gamma,\epsilon) $ can  be recovered 
as follows: If $\gamma_1$ and $\gamma_2$ are two paths such that the end-point of $\gamma_2$ 
coincides with the initial point of $\gamma_1$, we denote by $\gamma_1\gamma_2$ the path obtained 
by first traversing $\gamma_2$ and then $\gamma_1$. We then have the relation:
\beq
\mathbb{W}(\gamma_1\gamma_2,\epsilon) = \mathbb{W}(\gamma_1,\epsilon)\mathbb{W}(\gamma_2,\epsilon)\,.
\eeq
\end{rem}

Let us conclude this discussion by briefly mentioning how we can fix  the initial conditions. 
Clearly, one way is to know the value of the integral in one point. 
It is often useful to consider the value of the integral at some singular point, 
for instance where some scales vanish and the integral is expected to simplify.
If the differential equation is in canonical form, then all singularities are logarithmic. 
For example, close to the singularity described by eq.~\eqref{eq:log_singularity}, 
the differential equation takes the form
\beq\label{eq:singPointExp}
\rd \vec{\mathcal{J}}(\xi,\eps) \sim \eps\,\dlog p(\xi) \,A_p\,\vec{\mathcal{J}}(\xi,\eps)+\cdots\,,
\eeq
where $A_p$ is a constant matrix and the dots indicate terms that are power-suppressed 
in the limit $p(\xi)\to 0$. The solution to this equation is simply
\beq\label{eq:singPointSol}
\vec{\mathcal{J}}(\xi,\eps)\sim p(\xi)^{\eps A_p}\,\vec{\mathcal{J}}_p(\eps) +\cdots= R\,p(\xi)^{\eps J}\,R^{-1}\,\vec{\mathcal{J}}_p(\eps)+\cdots\,,
\eeq
where in the last step we have inserted the Jordan decomposition for the matrix $A_p = RJR^{-1}$ and $\vec{\mathcal{J}}_p(\eps)$ is related to the value of $\vec{\mathcal{J}}(\xi,\eps)$ in the limit $p(\xi)\to0$. The matrix $J$ is block diagonal:
\beq
J = \left(\begin{matrix}
J_1 & 0 & \cdots &0\\
0 & J_2 & \cdots &0\\
\vdots&&\ddots &\vdots\\
0&\cdots &0& J_r
\end{matrix}\right)\,, \qquad J_i = \underbrace{\left(\begin{matrix}
\lambda_i & 1 & 0 &\cdots &0\\
0 & \lambda_i &1& \cdots &0\\
\vdots &\vdots&&\ddots& \vdots\\
0&0&0&\cdots &\lambda_i
\end{matrix}\right)}_{r_i}\,,
\eeq
where the $\lambda_i$ are the eigenvalues of the matrix $A_p$ and $r_i$ is the size of the $i^{\textrm{th}}$ Jordan block.
The matrix exponential is then easy to carry out for each Jordan block $J_i$:
\beq
p(\xi)^{\eps J_i} = p(\xi)^{\lambda_i\eps}\,\left(\begin{matrix}
1 & \frac{\eps}{1!}\log p(\xi)& \frac{\eps^2}{2!}\log^2 p(\xi)& \cdots &\frac{\eps^{r_i-1}}{(r_i-1)!}\log^{r_i-1} p(\xi) \\
0&1 & \frac{\eps}{1!}\log p(\xi) & \cdots & \frac{\eps^{r_i-2}}{(r_i-2)!}\log^{r_i-2} p(\xi)  \\
&\vdots&&\ddots&\vdots\\
0&0&0&\cdots& 1
\end{matrix}\right)\,.
\eeq
We see that in a canonical basis we can easily predict the leading-power logarithmic behaviour close to the singularities. 
The asymptotic behaviour of Feynman integrals can also be analysed with independent techniques, 
like Mellin-Barnes integrals or the Strategy of Regions~\cite{Beneke:1997zp,Smirnov:2002pj,Pak:2010pt,Semenova:2018cwy}. 
By matching the two perspectives, we can fix the initial condition at the singular point. 
We will not say more about this approach here, 
and refer instead to the literature (cf.,~e.g.,~refs.~\cite{Henn:2013nsa,Dulat:2014mda}), 
because in many cases the initial condition can be fixed purely from physics input, 
without the need to explicitly evaluate the integrals in a limit. 
The idea is that the differential equation exposes \emph{all} the singularities that its solutions may have. 
The Feynman integrals, however, are very special solutions, and the structure of their singularities 
is very constrained from physical principles such as unitarity. In particular, a Feynman integral can only be singular 
if we go to a kinematic configuration where intermediate particles go on-shell. 
This implies that several singularities, present in the differential equation and therefore also in the general solution, 
must be absent in the particular solution we are looking for. 
It has been observed in many cases that by imposing this condition we are in fact able to determine
all the initial conditions (up to some simple one-scale integrals, that are simply a choice of normalisation
of the solution).

\begin{ex}[One-loop bubble with a single massive propagator]
    As an example of how to solve a differential equation in canonical form, we return to the example of the one-loop
    bubble integral. For simplicity, we fix $m^2_2=0$ and we will consider
    the canonical $\dlog$ basis $\vec{\mathcal{J}}(p^2,m_1^2;\epsilon)$ obtained by setting 
    $D=2-2\epsilon$ in \cref{eq:cob1mb}. We start by changing variables from $(p^2,m_1^2)$
    to $(u=p^2/m_1^2,m_1^2)$ and defining:
    \begin{equation}
        \vec{\mathcal{J}}(p^2,m_1^2;\epsilon)=(m_1^2)^{-\epsilon}\vec{\mathcal{J}}(u;\epsilon)\,.
    \end{equation}
    The differential equation satisfied by $\vec{\mathcal{J}}(u;\epsilon)$ is found to be
    \begin{equation}\label{eq:diffEqPureBub}
    \partial_{u}\,\vec{\mathcal{J}}(u;\epsilon)
    =\epsilon\left[
    \begin{pmatrix}
    -2&0\\
    0&0
    \end{pmatrix}\dlog(1-u)+
    \begin{pmatrix}
    1&-1\\
    0&0
    \end{pmatrix}\dlog u
    \right]\vec{\mathcal{J}}(u;\epsilon)\,,
    \end{equation}
    where we have made it explicit that it is in canonical $\dlog$ form.
    Since the single-scale tadpole integral only depends on $m_1^2$,
    this differential equation does not constrain ${\mathcal{J}}_2(u;\epsilon)$. This integral
    is considered trivial compared to the bubble, and was already quoted in \cref{eq:tad}.
    Adjusting the normalisation according to the basis in  \cref{eq:cob1mb}, it is:
    \begin{equation}\label{eq:J2_tad}
        {\mathcal{J}}_2(u;\epsilon)=e^{\gamma_E\epsilon}\Gamma(1+\epsilon)\,.
    \end{equation}
    Let us now focus our attention on solving the differential equation \eqref{eq:diffEqPureBub}
    order by order in $\epsilon$, and let ${\mathcal{J}}_{i,u_0}(\epsilon)$ be the initial condition for ${\mathcal{J}}_{i}(u;\epsilon)$ at $u=u_0$.  Equation~\eqref{eq:J2_tad} then implies $ {\mathcal{J}}_{2,u_0}(\epsilon)= e^{\gamma_E\epsilon}\Gamma(1+\epsilon)$, for all values of $u_0$. We define
    \begin{equation}
        {\mathcal{J}}_{i,u_0}(\epsilon)=\sum_{k=0}^\infty\eps^k\,{\mathcal{J}}^{(k)}_{i,u_0}\,.
    \end{equation} 
    The first two orders of $ {\mathcal{J}}_1(u;\epsilon)$ are:
    \begin{equation}\label{eq:solW1}
        {\mathcal{J}}_1(u;\epsilon)={\mathcal{J}}^{(0)}_{1,u_0}
        +\epsilon\left[ {\mathcal{J}}^{(1)}_{1,u_0}+
        \left(\mathcal{J}^{(0)}_{1,u_0}-{\mathcal{J}}^{(0)}_{2,u_0}\right)\int_{u_0}^u\frac{\rd t}{t}
        -2\,\mathcal{J}^{(0)}_{1,u_0}\int_{u_0}^u\frac{\rd t}{1-t}
        \right]+\mathcal{O}\left(\epsilon^2\right)\,.
    \end{equation} 
    Let us now note that \cref{eq:diffEqPureBub}  is singular
    at $u=0$, corresponding to $p^2=0$, hence any general solution to \cref{eq:diffEqPureBub} should have
    a logarithmic divergence at that point. This is clearly the case in \cref{eq:solW1}. The one-loop
    bubble, however, should only have singularities at the thresholds $m_1^2=0$ and $p^2=m_1^2$
    (see, e.g.,~ref.~\cite{SMatrix} for a textbook discussion of the singularities of Feynman integrals, and
    refs.~\cite{Gaiotto:2011dt,Abreu:2015zaa} for a more modern perspective that applies to
    Feynman integrals with massive propagators). It is clear that this implies that 
    $\mathcal{J}^{(0)}_{1,u_0}={\mathcal{J}}^{(0)}_{2,u_0}$, that is the leading term of the
    bubble integral should be equal to the leading term of the tadpole integral, which we have already
    computed above. It is in fact particularly convenient to choose $u_0=0$,
    and we will now use the strategy outlined above to determine
    the boundary condition to all orders in $\epsilon$
    (since this is a singular point of the differential equation,
    we refer the reader to \cref{sec:it_regularisation}, in particular \cref{eq:exampleRef},
    for a discussion on the regularisation of the integrals). In the neighbourhood of $u=0$
    we can expand the differential equation as in \cref{eq:singPointExp}:
    \begin{equation}
    \partial_{u}\,\vec{\mathcal{J}}(u;\epsilon)
    \sim\dlog u\begin{pmatrix}
    1&-1\\
    0&0
    \end{pmatrix}\vec{\mathcal{J}}(u;\epsilon)+\ldots\,.
    \end{equation}
    The solution is of the form of \cref{eq:singPointSol}, that is
    \begin{equation}
       \vec{\mathcal{J}}(u;\epsilon)
        \sim R\,u^{\epsilon J}\,R\,\begin{pmatrix}
        \mathcal{J}_{1,0}(\epsilon)\\
        \mathcal{J}_{2,0}(\epsilon)
        \end{pmatrix}+\ldots
        =\begin{pmatrix}
        \left(\mathcal{J}_{1,0}(\epsilon)-\mathcal{J}_{2,0}(\epsilon)\right)
        u^\epsilon+\mathcal{J}_{2,0}(\epsilon)\\
        \mathcal{J}_{2,0}(\epsilon)
        \end{pmatrix}+\ldots\,,
    \end{equation}
    where
    \begin{equation}
        R=\begin{pmatrix}
        1&1\\
        1&0
        \end{pmatrix}\,,\qquad
        J=\begin{pmatrix}
        0&0\\
        0&1
        \end{pmatrix}\,.
    \end{equation}
    The requirement that $\vec{\mathcal{J}}(u;\epsilon)$ should be
    regular at $u=0$ order by order in the $\eps$-expansion implies
    \begin{equation}
        \mathcal{J}_{1,0}(\epsilon)=\mathcal{J}_{2,0}(\epsilon) = e^{\gamma_E\epsilon}\Gamma(1+\epsilon)\,.
    \end{equation}
    The solution for the bubble integral is then:
    \begin{equation}
        {\mathcal{J}}_1(u;\epsilon)=1
        -2\,\epsilon\log(1-u)+\epsilon^2\left(\frac{\pi^2}{12}+2\log^2(1-u)+2\,\textrm{\emph{Li}}_2(u)\right)
        +\mathcal{O}\left(\epsilon^3\right)\,.
    \end{equation}
    After accounting for the different normalisation, 
    this expression is found to agree with the expansion of the all-order solution in \cref{eq:oneMBubAllOrder}.
\end{ex}

\subsubsection{Finding a canonical form}
The discussion of the previous sections makes it clear that it is of great 
advantage to bring a system of differential equations into canonical form. 
The existence of a canonical basis is however conjectural, and there is no 
general algorithm to find a transformation that brings a given system into 
canonical form. In fact, it is only known how to find a canonical basis for an 
arbitrary number of external legs and for arbitrary propagator masses in the 
case of one-loop integrals~\cite{Abreu:2017enx,Abreu:2017mtm,
Arkani-Hamed:2017ahv,Caron-Huot:2021xqj,Chen:2022fyw}.

Over the last couple of years substantial progress has been made in 
understanding how to find a canonical $\dlog$-form, if it exists. We have 
already argued that a necessary condition for a canonical $\dlog$-form to 
exist is that all maximal cuts at $\epsilon=0$ evaluate to algebraic 
functions. This is usually equivalent to requiring that the loop integrand can be 
completely localised via a sequence of residues. The algebraic 
functions obtained in this way are closely related to the maximal cuts of 
the integral at $\eps=0$, and they are often referred to as 
\emph{leading singularities}~\cite{Cachazo:2008vp}. It was conjectured in 
ref.~\cite{ArkaniHamed:2010gh} that an integral evaluates to a pure function 
if an only if it can be normalised in such a way that all its non-vanishing 
leading singularities are $\pm1$. This observation was at the heart of the 
original conjecture of ref.~\cite{Henn:2013pwa}, which states that a basis 
is canonical if and only if it has unit leading singularities. A very 
important strategy for finding a canonical $\dlog$-form therefore consists 
in finding a basis of master integrals with unit leading singularity. 
Closely related to integrals with unit leading singularity are Feynman 
integrals whose integrals can be written in terms of $\dlog$-forms with 
algebraic arguments, cf.~ref.~\cite{Henn:2020lye}. Since leading 
singularities are also closely related to maximal cuts, which themselves solve
the associated homogeneous system of differential equations 
(cf.~section~\ref{sec:solve_deq}), many strategies to finding a canonical 
basis start by solving these homogeneous equations 
(cf.,~e.g.,~\cite{Gehrmann:2014bfa}). There are several public computer 
codes that can be used  to find a canonical basis, for example 
\texttt{Fuchsia}~\cite{Gituliar:2017vzm}, 
\texttt{Canonica}~\cite{Meyer:2017joq}, \texttt{Libra}~\cite{Lee:2020zfb} 
and~\texttt{Initial}~\cite{Dlapa:2020cwj}, all of which are tailored
to slightly different situations. A method to find a canonical 
basis based on the computation of intersection numbers 
(see section~\ref{sec:intersection}) has also been 
proposed~\cite{Chen:2022fyw,Chen:2022lzr}. For large systems of integrals depending on
many scales, we find that a combination of all approaches is often
necessary to find a canonical basis, see e.g.~refs.~\cite{Abreu:2018rcw,Abreu:2018aqd,Chicherin:2018old,Abreu:2020jxa,Abreu:2021smk}.

In cases where the maximal cuts cannot be evaluated in terms of algebraic 
functions for $\eps=0$, no canonical $\dlog$-form can exist. While several 
examples exist where it was possible to find a canonical form 
nonetheless~\cite{Adams:2018yfj,Adams:2018bsn,Adams:2018kez,Broedel:2018rwm,
Bogner:2019lfa}, there is no systematic understanding of how to do so. 
In the future, it would be interesting to 
clarify the connection between generalisations of the concept of leading 
singularity to such cases~\cite{Bourjaily:2020hjv} and the existence of the 
canonical basis.


\section{Iterated integrals}
\label{sec:iterated}

It follows from Propositions~\ref{prop:iterated} and~\ref{prop:pure} that the natural class of functions that arise from the Laurent-expansion of dimensionally-regulated Feynman integrals are iterated integrals. In this section we present this class of functions in detail.

\subsection{General definitions}
\label{sec:iterated_def}

In the following we consider a geometric space $X$, and we always fix a set of local coordinates $\xi = (\xi_1,\ldots,\xi_s)$. If $\gamma$ is a path in $X$, and $\omega_1,\ldots,\omega_n$ are one-forms on $X$, we have already defined the iterated integral of $\omega_1\cdots\omega_n$ along $\gamma$ in eq.~\eqref{eq:iterated_def}. In the following it will be useful to refer to the one-forms $\omega_i$ as \emph{letters}
 and to $\omega_1\cdots\omega_n$ as a \emph{word} of length $n$. 
The set of all (independent) letters is called the \emph{alphabet}.
It will be useful to consider linear combinations of words
(unless stated otherwise, we always consider linear combinations with rational numbers as coefficients), and integration is linear:
\beq
\int_\gamma\left(\alpha\,\omega_1\cdots\omega_n + \beta\, \omega_1'\cdots\omega_m'\right) 
= \alpha\int_{\gamma}\omega_1\cdots\omega_n  +\beta\int_{\gamma}\omega_1'\cdots\omega_m'\,.
\eeq
By convention, it is useful to define the integral of the empty word to be $\int_{\gamma}()=1$.

Iterated integrals satisfy several general properties~\cite{ChenSymbol}:
\begin{enumerate}
\item \underline{Shuffle product:}
\beq
\int_{\gamma}\omega_1\cdots\omega_n\cdot \int_{\gamma}\omega_1'\cdots\omega_m' = \int_{\gamma}(\omega_1\cdots\omega_n)\shuffle (\omega_1'\cdots\omega_m')\,,
\eeq
where in the right-hand side we have introduced the shuffle product on words, defined recursively by
\beq\begin{split}
(\omega_1\cdots\omega_n)&\shuffle (\omega_1'\cdots\omega_m') \\
&\,= \omega_1\big((\omega_2\cdots\omega_n)\shuffle (\omega_1'\cdots\omega_m')\big) +\omega_1'\big((\omega_1\cdots\omega_n)\shuffle (\omega_2'\cdots\omega_m')\big)\,,
\end{split}\eeq
and $(\omega_1\cdots\omega_n)\shuffle () = ()\shuffle(\omega_1\cdots\omega_n) = (\omega_1\cdots\omega_n)$.
\item \underline{Path composition:} 
\beq\label{eq:path_composition}
\int_{\gamma_1\gamma_2}\omega_1\cdots\omega_n = \sum_{k=0}^n\int_{\gamma_1}\omega_1\cdots\omega_k\cdot \int_{\gamma_2}\omega_{k+1}\cdots\omega_n\,.
\eeq
\item \underline{Path reversal:} 
\beq
\int_{\gamma^{-1}}\omega_1\cdots\omega_n = (-1)^n \int_{\gamma}\omega_n\cdots\omega_1\,,
\eeq
where $\gamma^{-1}$ denotes the path $\gamma$ traversed in the opposite direction.
\end{enumerate}


\subsection{Homotopy invariance}

We have seen in section~\ref{sec:diffeqs} that whenever the matrix $A$ satisfies the integrability condition in eq.~\eqref{eq:DEQ_integrability}, then the solution does not depend on the details of the path $\gamma$, but only on its endpoints (as long as we do not cross any singularity). As a consequence, the iterated integrals that arise in the $\eps$-expansion should have the same property.
The iterated integrals defined in section~\ref{sec:iterated_def}, however, will in general depend on the details of the path $\gamma$, and not just on its endpoints. For example, consider the case $X = \mathbb{C}^2$ with coordinates $\xi = (\xi_1,\xi_2)$ and $\omega_i=\dlog \xi_i$. We consider the line segments
\beq\begin{split}
\gamma_1(t) = (1+(\xi_{10}-1)t,1)\,, &\quad \gamma_2(t) = (\xi_{10},1+(\xi_{20}-1)t)\,,\\
\gamma_3(t) =(1,1+(\xi_{20}-1)t)\,, &\quad \gamma_4(t) = (1+(\xi_{10}-1)t,\xi_{20})\,,
\end{split}\eeq
with $0\le t\le 1$.
Then $\gamma_{12}=\gamma_1\gamma_2$ and  $\gamma_{34}=\gamma_3\gamma_4$ are two paths from the point $(1,1)$ to the point $(\xi_{10},\xi_{20})$. These two paths clearly have the same endpoints, but the iterated integrals of the word $\omega_1\omega_2$ depend on the details of the paths. Indeed, using the path composition formula~\eqref{eq:path_composition}, one finds:
\beq
\int_{\gamma_{12}}\omega_1\omega_2= \log \xi_{10}\,\log\xi_{20} \textrm{~~~and~~~} \int_{\gamma_{34}}\omega_1\omega_2 = 0\,.
\eeq
If instead we consider the linear combination of words $w=\omega_1\omega_2+\omega_2\omega_1 = \omega_1\shuffle\omega_2$, then we have
\beq
\int_{\gamma_{12}} w = \int_{\gamma_{34}} w =\log \xi_{10}\,\log\xi_{20}\,.
\eeq
This illustrates the fact that the $\eps$-expansion of the path-ordered exponential can only furnish special linear combinations of words such that the iterated integrals only depend on the endpoints of the path, provided that the differential-equation matrix satisfies the integrability condition in eq.~\eqref{eq:DEQ_integrability}.
In the remainder of this section we describe a necessary and sufficient condition that these special linear combinations need to satisfy.

Two paths $\gamma_1$ and $\gamma_2$ that can be continuously deformed into each other while keeping the endpoints fixed are called \emph{homotopic}. Homotopy defines an equivalence relation on paths, and the equivalence classes are called \emph{homotopy classes}. An (iterated) integral that does not depend on the details of the path, but only on its endpoints, is called \emph{homotopy invariant}. The iterated integrals that arise from the $\eps$-expansion of the path-ordered exponential are homotopy invariant, provided that the differential-equation matrix satisfies the integrability condition~\eqref{eq:DEQ_integrability}.

In order to understand the criteria for homotopy invariance, let us start by understanding the case of length $n=1$. Consider a one-form $\omega_1$ and two homotopic paths $\gamma_1$  and $\gamma_2$ with the same endpoints. It is easy to see that the integrals of $\omega_1$ along the two paths agree if and only if
\beq
\int_{\gamma_1\gamma_2^{-1}}\omega_1 = 0\,.
\eeq
It is clear that $\gamma_1\gamma_2^{-1}$ is a closed curve. 
Consider now a surface $D$ whose boundary is $\partial D = \gamma_1\gamma_2^{-1}$. Stokes' Theorem implies
\beq
0 = \int_{\gamma_1\gamma_2^{-1}}\omega_1 = \int_{\partial D}\omega_1 = \int_{D}\rd \omega_1\,.
\eeq
If we want this  identity to hold for all paths, we see that  the integral of $\omega_1$ along a path is homotopy invariant if and only if $\omega_1$ is closed, $\rd\omega_1=0$.
Hence, for $n=1$, the criterion for homotopy invariance reduces to requiring that the form is closed, i.e., it is annihilated by the total differential. This criterion can in fact be generalised to higher length~\cite{ChenSymbol}: if $w$ is a linear combination of words, then the iterated integrals of $w$ are homotopy invariant, and thus independent of the details of the path, if and only if $Dw=0$, where $D$ is the differential that acts on words of one-forms $\omega_1\cdots\omega_n$ via
\beq\label{eq:bar_differential}
D(\omega_1\cdots\omega_n) = \sum_{i=1}^n\omega_1\cdots(\rd\omega_i)\cdots\omega_n + \sum_{i=1}^{n-1}\omega_1\cdots(\omega_i\wedge\omega_{i+1})\cdots\omega_n\,.
\eeq
A linear combination of words $w$ that is annihilated by this differential is 
called \emph{integrable}. We note that this is just another incarnation of the
integrability condition of \cref{eq:DEQ_integrability}. 
This discussion leads us to the following conclusion~\cite{ChenSymbol}:
\begin{prop}
An iterated integral is homotopy invariant if and only if the linear combination of words is integrable.
\end{prop}

\begin{rem}
In the case where all one-forms in an integrable word $w$ are $\dlog$-forms, $\omega_i=\dlog a_i(\xi)$ for some algebraic functions $a_i(\xi)$, then this integrable word can be identified with the \emph{symbol} of the iterated integral $\int_{\gamma}w$, which has prominently appeared in the context of Feynman integrals and scattering amplitudes, cf. refs.~\cite{ChenSymbol,Goncharov:2009tja,Goncharov:2010jf,Brown:2009qja,Duhr:2011zq}. In that context the word $\omega_1\cdots\omega_n$ is usually written $a_1(\xi)\otimes\cdots\otimes a_n(\xi)$.
\end{rem}

\begin{rem}
If all the one-forms $\omega_i$ are closed, $\rd\omega_i=0$, then eq.~\eqref{eq:bar_differential} takes the simpler form:
\beq
D(\omega_1\cdots\omega_n) =  \sum_{i=1}^{n-1}\omega_1\cdots(\omega_i\wedge\omega_{i+1})\cdots\omega_n\,.
\eeq
This is in particular the case when all one-forms are $\dlog$-form, in which case the condition $Dw=0$ reduces to the well-known \emph{integrability condition for the symbol~$w$}, see for example ref.~\cite{Gaiotto:2011dt}. More generally, this is the case for iterated integrals that arise from differential equations in canonical from, cf.~eq.~\eqref{eq:DEQ_integrability_can}.
A special case is when $X$ is one-dimensional and all one-forms $\omega_i$ are holomorphic. Then not only are all one-forms on $X$ automatically closed, but we also have $\omega_i\wedge \omega_{i+1}=0$. Hence, if $X$ is one-dimensional, all words are integrable, and so all iterated integrals are homotopy invariant.
\end{rem}


\subsection{Regularisation}
\label{sec:it_regularisation}
We have argued that in applications it might be useful to choose a singular point to fix the initial condition of a system of (canonical) differential equations for Feynman integrals. However, if an endpoint of the path $\gamma$ is a singular point, then the iterated integrals arising from the expansion of the path-ordered exponential will typically be divergent. In this scenario, we need to replace the integrals that arise from expanding the path-ordered exponentials by a suitably regularised version. The regularisation should satisfy certain conditions. For example, it should agree with the naive definition of iterated integrals in section~\ref{sec:iterated_def} whenever the integral converges, and all algebraic properties (like the shuffle product, path composition formula, etc.) should apply also to the regularised version. In the remainder of this section we present such a regularisation, sometimes called \emph{shuffle-regularisation} or \emph{tangential base-point regularisation} (cf.,~e.g.,~ref.~\cite{DeligneTangential}).

We restrict the discussion to the case where all singularities are logarithmic (which is automatically the case for systems of differential equations in canonical form) and where the space $X$ is one-dimensional. This case is typically sufficient for applications, as we may always choose a path that is piecewise constant, and on each segment we can reduce the problem to a one-dimensional case. We therefore assume from now on, without loss of generality, that the path $\gamma$ goes from the origin $\xi=0$ to the point $\xi=x$, and some of the one-forms $\omega_i$ have a logarithmic singularity at the origin,
\beq\label{eq:omega_infty}
\omega_i = a_i\,\dlog \xi+\ldots\,, \quad a_i\in\mathbb{C}\,.
\eeq
We also assume that there is no other singularity on the integration contour.
The regularised value of $\int_{0}^x\omega_1\cdots \omega_n = \int_{\gamma}\omega_1\cdots \omega_n$ is then defined as follows:
\begin{enumerate}
\item We introduce a small cut-off $\varepsilon$, i.e., we replace the integral $\int_{0}^x\omega_1\cdots \omega_n$ by $\int_{\varepsilon}^x\omega_1\cdots \omega_n$. This integral is convergent for all $\varepsilon\neq 0$.
\item Since all singularities are logarithmic, the integral behaves in the limit $\varepsilon\to 0$ as
\beq
\int_{\varepsilon}^x\omega_1\cdots \omega_n = \sum_{k=0}^n I_k(x)\log^k\varepsilon+\mathcal{O}(\varepsilon)\,.
\eeq
\item The regularised value of $\int_{0}^x\omega_1\cdots \omega_n$ is then defined to be the constant term $I_0(x)$.
\end{enumerate}
It is easy to check that this regularisation satisfies all our requirements:
\begin{itemize}
\item If the original integral is convergent, then it agrees with the regularised version:
\beq
\int_{0}^x\omega_1\cdots \omega_n = \lim_{\varepsilon\to0} \int_{\varepsilon}^x\omega_1\cdots \omega_n = I_0(x)\,.
\eeq
\item The regularisation is consistent with the shuffle product, etc. We can perform all algebraic manipulations in a naive way, without having to worry about the regularisation. Indeed, $ \int_{\varepsilon}^x\omega_1\cdots \omega_n$ is convergent for $\varepsilon\neq0$, and so all naive manipulations apply. Moreover, the projection onto the constant term $I_0(x)$ is consistent with multiplication (the constant term of a product is the product of the constant terms), and so the regularisation and multiplication operations commute.
\end{itemize}

\begin{ex}\label{ex:log_reg}
Let us illustrate this on the example of the integral $\int_0^x\frac{\rd\xi}{\xi}$. Clearly, this integral is divergent. We introduce a cut-off, and consider the integral $\int_{\varepsilon}^x\frac{\rd\xi}{\xi} = \log x - \log\varepsilon$. The regularised value of the integral is then $\log x = \int_{1}^x\frac{\rd\xi}{\xi}$.
\end{ex}

There is one important issue we need to address: The regularised value obtained in this way depends on a `choice of regularisation scheme'. Indeed, we are free to rescale the cut-off by some non zero constant, $\varepsilon\to v\,\varepsilon$, and the regularised value will depend on the choice of $v$. This rescaling factor can be interpreted as the angle and speed of approach to the singularity. It is therefore natural to identify ${v}$ with a tangent vector of $\gamma$ at the point $x_0=0$ (called a \emph{tangential base-point}), and we will denote it by $\vec{v}_{x_0}=\vec{v}_0$. The regularised value $\int_{0}^x\omega_1\cdots \omega_n$ with the tangential base-point $\vec{v}_0$ at the origin is then denoted by $\int_{\vec{v}_0}^x\omega_1\cdots \omega_n$. Note that, if $\int_{0}^x\omega_1\cdots \omega_n$ converges, its value is independent of the choice of tangent vector.

\begin{ex}
Let us return to Example~\ref{ex:log_reg}. If we choose as a cut-off $v\,\varepsilon$, we have $\int_{v\,\varepsilon}^x\frac{\rd\xi}{\xi} = \log x-\log v - \log\varepsilon$. Hence, the shuffle-regularised value with respect to the tangent vector $\vec{v}_0$ is
\beq
\int_{\vec v_0}^x\frac{\rd \xi}{\xi} = \log x - \log v\,.
\eeq
\end{ex}

\begin{rem}
The previous discussion implies that whenever we want to fix the initial condition of a system of differential equations at a singular point, then we have to interpret the iterated integrals that arise from the $\eps$-expansion of the path-ordered exponential as their regularised versions. This regularisation then depends on a `scheme-choice', namely the choice of the tangent vector $\vec v_{x_0}$ at the initial point $x_0$, and the Laurent coefficients of the path-ordered exponential will depend explicitly on $\log v$. The final result for the integral can of course not depend on the scheme choice: the initial condition $\vec{\mathcal{J}}_0(\eps)$ will also depend on the choice of $v$ in such a way that the dependence cancels. It is
customary to choose $v=1$, so that no explicit logarithms of $v$ appear.
\end{rem}

The regularisation we just described seems to be very hard to implement in practice, because it requires one to introduce a cut-off, expand in the regulator $\varepsilon$ and then project onto the constant term. It is actually possible to obtain a closed formula for the shuffle-regularised version. We only present here the result, and refer to ref.~\cite{Brown:mmv}, section~4, for details. It is useful to define $\omega_i^{\infty} = a_i\,\dlog\xi$, so that $\omega_i = \omega_i^{\infty} + \cdots$  (cf.~eq.~\eqref{eq:omega_infty}). We then have~\cite{Brown:mmv}
\beq\begin{split}\label{eq:Brown_R}
\int_{\vec{v}_0}^x\omega_1\cdots \omega_n &\,= \sum_{k=0}^n \int_{\vec{v}_0}^x \omega_1^{\infty}\cdots \omega_k^{\infty}\, \int_0^xR[\omega_{k+1}\cdots\omega_n]\\
&\,= \sum_{k=0}^n \frac{a_1\cdots a_k}{k!}\,\log^k\frac{x}{v}\,\int_0^xR[\omega_{k+1}\cdots\omega_n]\,,
\end{split}\eeq
where the map $R$ is defined by
\beq
R[\omega_{1}\cdots\omega_n] = \sum_{k=0}^n(-1)^k\,( \omega_k^{\infty}\cdots \omega_1^{\infty})\shuffle (\omega_{k+1}\cdots\omega_n)\,.
\eeq
The map $R$ replaces the word $\omega_{1}\cdots\omega_n$ by a linear combination of words so that their iterated integral is convergent over the whole range $[0,x]$. 

\begin{rem} Equation~\eqref{eq:Brown_R} takes a very simple form if there is a single one-form that diverges at $\xi=0$. In that case eq.~\eqref{eq:Brown_R} is equivalent to unshuffling all the occurences of this one-form. For example, for $\omega_1=\dlog \xi$ and $\omega_2 = \dlog (1-\xi)$, eq.~\eqref{eq:Brown_R} is equivalent to the well known shuffle-regularisation formula:
\beq\begin{split}\label{eq:exampleRef}
\int_{\vec{v}_0}^{x}\omega_1\omega_2 &\,= \int_{\vec{v}_0}^{x}\omega_1 \int_{\vec{v}_0}^{x}\omega_2 - \int_{\vec{v}_0}^{x}\omega_2\omega_1\\
&\,= \log\frac{x}{v}\int_{0}^{x}\omega_1  - \int_{0}^{x}\omega_2\omega_1\\
&\,= \log\frac{x}{v} \log(1-x)  + \textrm{Li}_2(x)\,,
\end{split}\eeq
where $\textrm{Li}_2(x)$ is the dilogarithm function, defined for $|x|<1$ by
\beq
\textrm{Li}_2(x) = \sum_{n=1}^{\infty}\frac{x^n}{n^2} = -\int_0^x\frac{\rd \xi}{\xi}\,\log(1-\xi)\,.
\eeq
\end{rem}

\begin{rem} The discussion in this section only applies to logarithmic singularities, which is sufficient to cover systems of differential equations in canonical form. If also one-forms with higher-order poles are present, the regularisation becomes more involved. For a discussion of this more general case, see ref.~\cite{matthes2021iterated}.
\end{rem}


\subsection{Linear independence of iterated integrals}
In applications one is often interested in knowing that the special functions introduced to express the answer are (linearly) independent. Indeed, working with an independent set of objects often leads to shorter analytic expressions that are free of hidden cancellations. The goal of this section is to state a linear independence result for iterated integrals (in particular those from systems of differential equations in canonical form). The mathematical background can be found in refs.~\cite{ChenSymbol,DDMS}.

Let us consider a system with $A(x,\eps) = \eps\,\widetilde{A}(x) = \eps\sum_i\widetilde{A}_i\omega_i$ (we follow the notation of section~\ref{sec:diffeqs}, but we do not explicitly require the $\omega_i$ to have only logarithmic singularities for the purpose of this section). At every order of the $\eps$-expansion, the path-ordered exponential will involve iterated integrals of words in the one-forms $\omega_i$. Clearly, if we work with arbitrary one-forms $\omega_i$, these iterated integrals will not be independent. For example, every linear relation between the $\omega_i$ will induce linear relations among the iterated integrals. Moreover, if a one-form is a total derivative, $\omega_i = \rd f$ for some function $f$, then we can integrate out $\omega_i$, again leading to relations among the iterated integrals. Here we have to make an important comment: of course, we can always locally find a primitive $f$ for every one-form $\omega$, and this $f$ will in general be a transcendental function. In applications, however, we are typically interested only in primitives taken from some restricted subalgebra $\mathcal{C}$ of functions such that in local coordinates $\omega_i = \sum_jc_{ji}(\xi)\,\rd\xi_j$ with $c_{ij}(\xi)\in\mathcal{C}$. For example, in the case of $\dlog$-forms $\omega_i=\dlog a_i(\xi)$ with algebraic $a_i(\xi)$, we would take $\mathcal{C}$ to be the field of algebraic functions in $\xi$ (in the case of non-$\dlog$-forms non-algebraic functions may also appear, such as modular forms). More generally, we will then say that the $\omega_i$ are linearly dependent with respect to $\mathcal{C}$ if there exists a function $f\in\mathcal{C}$ and constants $\alpha_i$ not all zero such that
\beq\label{eq:linear_dependence}
\sum_{i}\alpha_i\,\omega_i = \rd f\,.
\eeq
Linear independence with respect to $\mathcal{C}$ is defined in the obvious way, such that every identity of the form~\eqref{eq:linear_dependence} implies $\alpha_i=0$ for all $i$. 
It turns out that this linear independence of the one-forms is sufficient to guarantee that iterated integrals are independent \emph{as functions}.

\begin{prop}
The iterated integrals arising from the path-ordered exponential are linearly independent over $\mathcal{C}$ as functions if and only if the one-forms $\omega_i$ are linearly independent with respect to $\mathcal{C}$.
\end{prop}
The proof of this statement can be found in ref.~\cite{DDMS}.\footnote{Note that ref.~\cite{DDMS} strictly only considers the case of a single variable. It is not difficult to adjust the proof to more variables.}

\begin{rem} It is important to understand that we consider the linear independence of the iterated integrals as functions. If we consider evaluations at special values, there may be additional relations. For example, if $\omega_1=\dlog\xi$ and $\omega_2=\dlog(1-\xi)$, then the iterated integrals $\int_{\vec{1}_0}^x \omega_1\omega_2$ and $\int_{\vec{1}_0}^x \omega_2\omega_1$ are linearly independent as functions of $x$. If we evaluate these integrals at $x=1$, we find
\beq
\int_{\vec{1}_0}^1 \omega_1\omega_2 = \frac{\pi^2}{6} = - \int_{\vec{1}_0}^1 \omega_2\omega_1\,.
\eeq
\end{rem}


\subsection{Iterated integrals and Feynman integrals}
We have already seen that Propositions~\ref{prop:iterated} and~\ref{prop:pure} imply that iterated integrals naturally arise from differential equations satisfied by dimensionally-regulated Feynman integrals.
In applications it is often desirable to relate the results of a computation to known definitions of special functions which have been studied independently in the literature, e.g., there may be computer codes and algorithms for their manipulation and evaluation. There are several classes of iterated integrals that are well studied in mathematics and physics, which we briefly review in this section. Before doing so, however, we note that in some applications it is beneficial to consider classes of iterated integrals that are specific to a given kinematic configuration, such as pentagon functions for five-point Feynman integrals \cite{Gehrmann:2018yef,Chicherin:2020oor,Chicherin:2021dyp} or hexagon functions for specific six-point integrals \cite{Dixon:2013eka,Caron-Huot:2019bsq}, even if these very specialised iterated integrals can be related to more general classes of iterated integrals.

The arguably most prominent class of iterated integrals that arise from Feynman integral computations are \emph{multiple polylogarithms} (MPLs). They have appeared in mathematics a few hundred years ago (cf.,~e.g.,~refs.~\cite{Leibniz:1696,Abel_Oeuvres,Kummer,Lappo:1927}). Over the last decades, they resurfaced in both pure mathematics~\cite{Goncharov:2009tja,Brown:2009qja,Goncharov:2010jf,GoncharovMixedTate,Goncharov:1998kja} and physics~\cite{Remiddi:1999ew,Gehrmann:2000zt,Blumlein:2000hw,Blumlein:2009fz,Blumlein:2009ta}.
MPLs can be defined as the iterated integrals
\beq\label{eq:MPL_def}
G(a_1,\ldots,a_n;x) = \int_0^x\frac{\rd t}{t-a_1}\,G(a_2,\ldots,a_n;t)\,,
\qquad \textrm{with~} G(\,;x) \equiv 1.
\eeq
Note that if $a_n=0$, this integral is divergent, and we need to regularise it. This is typically done by introducing a tangential base-point $\vec{1}_0$ for the lower integration boundary. This is equivalent to the special definition
(see \cref{eq:Brown_R})
\beq
G(\underbrace{0,\ldots,0}_{n};x) = \frac{1}{n!}\,\log^nx\,.
\eeq
MPLs contain the ordinary logarithm and the classical polylogarithms as special cases, e.g., if $a\neq0$,
\beq
G(a;x) = \log\left(1-\frac{x}{a}\right) \textrm{~~and~~} G(\underbrace{0,\ldots,0}_{n-1},a;x) = -\textrm{Li}_n\left(\frac{x}{a}\right)\,.
\eeq
It is easy to check that MPLs are in fact the prime examples of pure functions (see section~\ref{sec:can_deq}).
The properties of MPLs, in particular their algebraic structure and the relations they satisfy, are well understood. An important tool in the study of MPLs is their so-called {symbol} and their coaction/coproduct, cf.~refs.~\cite{ChenSymbol,Goncharov:2009tja,Goncharov:2010jf,Brown:2009qja,Duhr:2011zq,Goncharov:2005sla,brownmixedZ,Brown:2011ik,Duhr:2012fh,Duhr:2014woa}. Moreover, there are several publicly available computer tools to manipulate and evaluate MPLs~\cite{Panzer:2014caa,Gehrmann:2001pz,Gehrmann:2001jv,Vollinga:2004sn,Ablinger:2009ovq,Buehler:2011ev,Ablinger:2013cf,Frellesvig:2016ske,Ablinger:2018sat,Duhr:2019tlz,Naterop:2019xaf,Wang:2021imw}.

It has been known from the early days of quantum field theory that not all Feynman integrals can be expressed in terms of MPLs. Indeed, A. Sabry discovered new functions of elliptic type in the calculation of the two-loop corrections
to the electron propagator in QED with massive electrons already in 1962~\cite{Sabry}. Since then, many other Feynman integrals have been identified that cannot be expressed in terms of MPLs~\cite{Aglietti:2007as,Mistlberger:2018etf,Adams:2018bsn,Adams:2018kez,Bogner:2019lfa,Broadhurst:1987ei,Bauberger:1994by,Bauberger:1994hx,Caffo:1998du,Laporta:2004rb,Kniehl:2005bc,Caffo:2008aw,Brown:2010bw,Muller-Stach:2011qkg,MullerStach:2011ru,CaronHuot:2012ab,Nandan:2013ip,Czakon:2013goa,Brown:2013hda,Remiddi:2013joa,Bloch:2013tra,Adams:2013kgc,Adams:2015gva,Adams:2015ydq,Adams:2014vja,Bloch:2016izu,Adams:2016xah,Remiddi:2016gno,vonManteuffel:2017hms,Broedel:2017siw}. There is a natural class of iterated integrals that generalise MPLs to functions of elliptic type, called elliptic multiple polylogarithms (eMPLs)~\cite{LevinRacinet,BrownLevin}. eMPLs have first appeared in physics in the context of one-loop scattering amplitudes in string theory~\cite{Broedel:2014vla,Broedel:2018izr,Broedel:2015hia,Broedel:2017jdo}, but they have also been used to express several multi-loop Feynman integrals that cannot be expressed in terms of eMPLs~\cite{Broedel:2019hyg,Broedel:2018qkq,Blumlein:2018jgc,Bezuglov:2020tff,Campert:2020yur,Kristensson:2021ani} (see refs.~\cite{Bloch:2013tra,Adams:2013kgc,Adams:2015gva,Adams:2015ydq,Ablinger:2017bjx,Blumlein:2018aeq} for alternative definitions of elliptic generalisations of polylogarithmic functions that are closely related to eMPLs). Closely related to eMPLs are iterated integrals of modular forms~\cite{Brown:mmv,ManinModular}, including meromorphic modular forms~\cite{matthes2021iterated}, which have also appeared in the context of Feynman integrals~\cite{Broedel:2018rwm,Adams:2017ejb,Broedel:2018iwv,Broedel:2019kmn,Abreu:2019fgk,Broedel:2021zij}. 
Iterated integrals of holomorphic modular forms and eMPLs are examples of pure functions~\cite{Broedel:2018qkq}, and the corresponding Feynman integrals satisfy differential equations in canonical form~\cite{Adams:2018yfj,Broedel:2018rwm,Bogner:2019lfa} (though it is not possible to bring them into canonical $\dlog$-form). The properties of these functions are not yet as well understood as in the case of MPLs, although substantial progress was made in recent years to understand their algebraic properties and numerical evaluation, see,~e.g.,~refs.~\cite{Bogner:2017vim,Duhr:2019rrs,Broedel:2019tlz,Walden:2020odh}.

MPLs, eMPLs and iterated integrals of modular forms are not the only classes of iterated integrals that have been observed to arise from Feynman integral computations. Other classes of special functions have been discovered~\cite{Lee:2019wwn,Adams:2018bsn,Adams:2018kez,Ablinger:2013cf,Remiddi:2016gno,Blumlein:2018jgc,Ablinger:2017bjx,Ablinger:2011te,Ablinger:2014bra,Remiddi:2017har,Chen:2017soz,Kniehl:2019vwr,Lee:2020obg,Lee:2020mvt,Bezuglov:2020ywm,Bezuglov:2021jou,Bezuglov:2021tax,Badger:2021owl,Kreimer:2022fxm}, even though for some of them it is known by now that they are related to the classes of functions discussed above. A completely new class of special functions has recently been discovered that is related to higher-dimensional generalisations of elliptic curves known as Calabi-Yau manifolds~\cite{Klemm:2019dbm,Bonisch:2020qmm,Bloch:2016izu,Bloch:2014qca,Bourjaily:2018ycu,Bourjaily:2018yfy,Bourjaily:2019hmc}, with new classes of iterated integrals which generalise the iterated integrals of modular forms to Calabi-Yau varieties of higher dimension~\cite{Bonisch:2021yfw}.

\paragraph{MPLs and systems in canonical $\dlog$-forms.}
Let us conclude this discussion of iterated integrals that arise from Feynman integrals by commenting on iterated integrals of $\dlog$-forms, i.e., iterated integrals of words of letters of the form $\dlog a_i(\xi)$, with $a_i(\xi)$ an algebraic function. So far in all cases relevant to Feynman integrals these algebraic functions only involve square roots.\footnote{It is  interesting to ask if  algebraic functions not involving square roots can also arise in Feynman integral computations.}

It is clear that iterated integrals of $\dlog$-forms are closely related to MPLs.
Indeed, if the $a_i(\xi)$ are linear functions, then we immediately reproduce the definition in eq.~\eqref{eq:MPL_def}. In many cases it is still possible to evaluate the resulting iterated integrals in terms of MPLs, even if the $a_i(\xi)$ are more general rational functions or even involve square roots.
\begin{itemize}
\item If all the $a_i(\xi)$ are rational functions of $\xi$, we can evaluate the iterated integrals in terms of MPLs in an algorithmic fashion. Indeed, we can assume without loss of generality that the $a_i(\xi)$ are polynomials. We can choose a piecewise constant path to evaluate the path-ordered exponential (cf. section~\ref{eq:solving_deqs}). On the segment $\gamma_{i_0}$ where all elements of $\xi$ but $\xi_{i_0}$ are constant, we can factor each polynomial $a_i(\xi)$ into linear factors in $\xi_{i_0}$, $a_i(\xi) = c_{i}\,\prod_{k}(\xi_{i_0}-\xi_{i_0,k})$, and we obtain only $\dlog$-forms depending on linear arguments:
\beq
\gamma_{i_0}^*\dlog a_i(\xi) = \sum_k\dlog(\xi_{i_0}-\xi_{i_0,k})\,.
\eeq
\item If the $a_i(\xi)$ involve square roots, the situation is different. We have to distinguish two cases. First, if we can find a change of variables $\xi = \varphi(\chi)$ that rationalises all square roots, then the functions $a_i(\varphi(\chi))$ are rational in $\chi$, and we reduce the problem to the previous case. In some cases it is possible to find such a change of variables in an algorithmic manner, or to show that is does not exist cf.,~e.g.,~refs.~\cite{Festi:2018qip,Besier:2018jen,Besier:2019hqd,Besier:2019kco,Besier:2020klg,Besier:2020hjf,Festi:2021tyq}. If not all square roots can be rationalised, then one cannot conclude a priori if a representation in terms of MPLs (with algebraic arguments) exists. Indeed, in ref.~\cite{Brown:2020rda} an explicit example of a double-iterated integral of $\dlog$-forms with a square root was constructed that \emph{cannot} be expressed in terms of MPLs. So far, however, all instances of systems in canonical $\dlog$-form with non-rationalisable square roots in the context of Feynman integrals can still be evaluated in terms MPLs, cf.~refs.~\cite{Heller:2019gkq,Kreer:2021sdt,Duhr:2021fhk}.
\end{itemize}


\section{Intersection theory for Feynman integrals}
\label{sec:intersection}

The previous sections have discussed the evaluation of Feynman integrals by fairly direct approaches. Recently there has been some interest in applications of the mathematics going by the name of intersection theory, which deals with certain multivalued integrals. One area of application is an alternative approach to computing IBP reduction coefficients, which might help to overcome some of the shortcomings of current approaches. 
Others involve an algebraic operation called the coaction, which can expose the behaviour of functions under differentiation or operations computing discontinuities.

In this section, we review the formalism of intersection theory for a class of integrals that includes Feynman integrals in dimensional regularisation, as well as the Euler-type hypergeometric integrals that arise upon their evaluation. 
A principal classic reference for the mathematics is ref.~\cite{AomotoKita}. For recent treatments including applications to physics, see refs.~\cite{Weinzierl:2022eaz,Mizera:2019gea,Mizera:2019ose,Cacciatori:2021nli}.

An integral is a pairing between two objects, the integrand and the contour of integration. Denoting these objects respectively by $\omega$ and $\gamma$, the integral is represented as 
\begin{equation}
\int_\gamma \omega.
\end{equation}
We will consider cases where
the differential form $\omega$ is a holomorphic differential form of degree $n$,
and the real dimension of the integration contour is $n$, within the complex projective space ${\mathbb{P}^{n}(\mathbb{C})}$. 

The integral remains unchanged if we add a surface term to the integrand, and $\omega$ is properly understood as a cohomology class, or cocycle. Likewise, the integration contour $\gamma$ is a homology class, or cycle. 
In the types of integrals with which we are concerned, the integrand involves multivalued functions. This can be seen in the parametric
representations of \cref{sec:paramRep}. For instance, 
let us consider the Baikov representation of a Feynman integral in eq.~\eqref{eq:Baikov}. The Baikov polynomial is raised to the power $\frac{D-K-1}{2} = n+\frac{\mu}{2}-\eps$, where $n$ is an integer and $\mu\in\{0,1\}$. 
Thus the integrands of Feynman integrals in dimensional regularisation are always multivalued functions.  Because of this multivaluedness, we need to work in the frameworks of {\em twisted} cohomology and homology, which we review below.

An example that we will keep in mind and follow through this section is the integral representation of the hypergeometric function ${}_2F_1$ defined in eq.~\eqref{eq:2f1def}.
While this function is not a Feynman integral by itself, some simple Feynman integrals such as the bubble studied in the previous sections are naturally expressed in terms of ${}_2F_1$ functions (see \cref{eq:oneMBubAllOrder,eq:twoMBubAllOrder}), and the additional prefactors do not affect the main points of intersection theory. Feynman integrals, as presented in \cref{eq:Feynman} and \cref{eq:Baikov}, can be viewed in terms of generalisations of this integral, i.e., polynomials in the kinematic invariants raised to complex powers. For the purpose of the present discussion, we will simplify this integral even further by dividing through by the gamma functions. Furthermore, we relabel the exponents to make their independence manifest. Thus we write our basic reference integral as
\begin{equation}\label{eq:2f1int}
	\int_0^1 u^{\alpha_0} (1-u)^{\alpha_1} (1-xu)^{\alpha_{1/x}} \,\rd u\,.
\end{equation}

\subsection{Twisted cohomology and twisted homology}
\paragraph{Twisted cohomology.}

We begin by considering the integrand $\omega$, which we take to be a holomorphic $n$-form
\begin{equation}\label{eq:factorizedIntegrand}
	\omega=\rd{u}\prod_{I} P_I({u})^{\alpha_I}\,,
\end{equation}
where $u=(u_1,u_2,\ldots,u_n)$ and  $\rd{u}=\rd u_1\wedge\ldots\wedge \rd u_n$,
where the $P_I$ are polynomials in the integration variables $u_i$ and additional kinematic variables $x_j$, 
and with complex exponents, $\alpha_I\in\mathbb{C}$. 
In the case of Feynman integrals 
in dimensional regularisation, we would assume that the exponents take the form
$\alpha_I=n_I+a_I\epsilon$, with $n_I\in\mathbb{Z}$,
$a_I\epsilon\in\mathbb{C}^*$, $\sum_I a_I \neq 0$, and where $\epsilon$ can be taken to be infinitesimally small. 
We restrict ourselves to the case where the coefficients in the Laurent expansion of the integral $\int_{\gamma}\omega$  can be given in terms of multiple polylogarithms, which requires the integer values of $n_I$ and some further restrictions on the form of the polynomials $P_I({u})$.

We decompose the integrand as 
$\omega=\Phi\varphi$, where $\Phi$ is a multi-valued function and $\varphi$ is a single-valued differential form,
\begin{equation}
\label{eq:Phiphi}
\Phi = \prod_{I} P_I({u})^{a_I \eps}\quad \textrm{and} \quad
\varphi = \rd{u} \prod_{I} P_I({u})^{n_I}\,.
\end{equation}

We define the \emph{twist} $\dlog\Phi$ 
and consider the covariant differential
\begin{equation}\label{eq:covd}
	\nabla_\Phi=\rd+\dlog\Phi\wedge\,.
\end{equation}
We then have $\rd(\Phi\xi) = \Phi\,\nabla_{\Phi}\xi$, where $\xi$ can be any smooth differential form. Stokes' Theorem implies that
for an arbitrary smooth $(n-1)$-form $\xi$ we have
\begin{equation}\label{eq:connection}
	\int_\gamma\Phi\varphi=
	\int_\gamma\Phi(\varphi+\nabla_\Phi\xi)\,.
\end{equation}
Thus the integrand is only defined up to adding a total covariant derivative, and 
 we are therefore considering elements of the (twisted) cohomology groups
\begin{equation}\label{eq:H_dR_def}
	H^n(X,\nabla_\Phi)=
	\{\varphi|\nabla_\Phi\varphi=0\}
	/
	\{\nabla_\Phi\xi\}\,.
\end{equation}
One might consider the twisted cohomology group  with the differential of \cref{eq:covd} acting on differential forms of different degrees, but a theorem of Aomoto \cite{Aomoto1975OnVO} states that $n$ is the only dimension with nonvanishing twisted cohomology.

Equation~\eqref{eq:connection} provides an alternative way to generate the IBP identities discussed in \cref{sec:ibps}.
The expansion of these covariant derivatives in terms of the original polynomials $P_I({u})$ has the effect of shifting the integer parts $n_I$ of the exponents in \cref{eq:factorizedIntegrand}, while leaving the complex part $a_I\eps$ unchanged. 
The space generated by the integrands with different values of $n_I$ gives the full space of twisted cohomology.

\begin{ex}
Let us return to our example of the integral in \cref{eq:2f1int}. Comparing \cref{eq:2f1int} with \cref{eq:factorizedIntegrand}, we see that $n=1$, $|I|=3$, and the set of polynomials can be taken to be $\{P_0=u,\,P_1=1-u,\,P_{1/x}=1-x u\}$. We have set $u_1=u$ for simplicity, not to be confused with the multi-index $u$ in general formulas such as \cref{eq:factorizedIntegrand}. 
Then $\varphi= u^{n_0} (1-u)^{n_1} (1-xu)^{n_{1/x}} \,\rd u$, $\Phi= u^{a_0\eps} (1-u)^{a_1\eps} (1-xu)^{a_{1/x}\eps}$, and
the twist 1-form is 
\begin{equation}\label{eq:2F1twist}
 \dlog\Phi = a_0\eps \frac{\rd u}{u} - a_1\eps \frac{\rd u}{1-u} - a_{1/x}\eps\frac{x\,\rd u}{1-x u}.
\end{equation}
The space of integrands modulo the IBP relations is known to be two-dimensional. We will discuss this fact along with a choice of basis shortly.
\end{ex}

\paragraph{Twisted homology.}

The integration contour $\gamma$
is a $n$-dimensional (relative) cycle in
\begin{equation}\label{eq:Xdef}
	X(\mathbb{C})={\mathbb{P}^{n}(\mathbb{C})}\setminus\bigcup\limits_{I} \{P_I({u})=0\}\,,
\end{equation}
Strictly speaking, $\{P_I({u})=P_I(u_1,\ldots,u_n)=0\}$ is an affine variety in $\mathbb{C}^n$. We use the same notation for the affine variety and its lift to projective space. In other words, $\gamma$ is a domain whose boundary is contained in the union of the varieties defined by $P_I({ u})=0$ (or equivalently $\Phi=0$). Note that this is in particular the case for the Baikov representation in eq.~\eqref{eq:Baikov}, where the integration cycle $\Delta$ is bounded by the variety where the Baikov polynomial vanishes.

Since $a_I\neq0$, then 
$\Phi$ vanishes on the boundary of $\gamma$, at least for some ranges of values of $\epsilon$, and thus for all values by analytic continuation. 
Therefore, if $a_I\neq0$, there are no boundary
contributions when performing integration by parts.

The concept of a twisted cycle, as an element of twisted homology, is essentially a version of the equivalence class of $\gamma$ that is taken together with a specific choice of branch of $\Phi$ to deal with the multi-valuedness. The ordinary Stokes' Theorem holds on the branch and corresponds to the version with the covariant derivative given above.

The effect of the twist is to replace the boundary of the cycle by a contour that encircles the singularity. The simplest twisted cycle is a loop around a singular point. Precise definitions and explanations may be found in section 2.3 of ref.~\cite{AomotoKita}.  When used as an integration contour, the result is the same as the physicists' notion of  the so-called `plus-prescription' that regularises the integral in the $\eps$ expansion (see e.g.~\cite{Plehn:2009nd}).

In the case of the integral in \cref{eq:2f1int}, we are given the contour $\gamma=[0,1]$, which is consistent with the description in \cref{eq:Xdef}. We will not write the twist explicitly, but it can be constructed with the knowledge of $\Phi$, and its use is implicit in regularising divergences if we perform a series expansion in $\eps$ before evaluating the integral. We note that there are various other cycles that can be constructed with boundaries contained in the union of varieties defined by $\Phi=0$, namely cycles whose endpoints belong to the set $\{0,1,1/x,\infty\}$. It turns out that at most two of them can be linearly independent (up to boundaries), as we discuss below in the context of a basis choice.

\paragraph{Bases of twisted cohomology and twisted homology.}

It will be useful to identify explicit bases for the twisted cohomology and homology groups associated to a particular integral. The problem of constructing  bases algorithmically appears to be difficult to solve. In practice, it is helpful to first identify the dimension of these groups, then select a set of elements of that cardinality, and finally test for linear independence. 
It is not obvious that the dimension of twisted cohomology should equal the dimension of twisted homology, but in the case where a universal coefficient theorem holds, there is a natural isomorphism between these groups, leading to the bilinear pairing expressed by integration (see Lemma 2.3 of ref.~\cite{AomotoKita}). 
For all cases of interest in this review, the dimensions of the two groups coincide.
Even in such cases, however, determining the dimension of the groups is not straightforward in general. 
There is an upper bound given by the number of 
critical points of the function $\Phi$, i.e., the 
number of independent solutions to the equation
\begin{equation}
	\dlog\Phi=0\,,
\end{equation}
and in many cases this bound is saturated, for example under conditions outlined in 
ref.~\cite{AomotoKita} (see also refs.~\cite{Lee:2013hzt,Bitoun:2018afx,Frellesvig:2019kgj}).

\begin{ex}
In the example of the integral in \cref{eq:2f1int}, we see from \cref{eq:2F1twist} that the equation $\dlog\Phi=0$ has two solutions generically, since we can factor out a quadratic polynomial in $u$ after multiplying the denominators through. Accordingly, the dimension of co/homology in this case is 2. To choose a basis of cohomology, we can try two different sets of integers $\{n_0,n_1,n_{1/x}\}$ in
$\varphi= u^{n_0} (1-u)^{n_1} (1-xu)^{n_{1/x}} \,\rd u$ and check their linear independence as cohomology classes. To choose a basis of homology, we can choose two sets of endpoints among $\{0,1,1/x,\infty\}$. In order to follow this example through the following sections, let us make a concrete choice and take
\beq\label{eq:2F1_contours}
\gamma_1=[0,1]\qquad \textrm{and} \qquad\gamma_2=[0,1/x]\,,
\eeq 
and
\begin{equation}\label{eq:canonicForm2F1}
	\varphi_1=\dlog\frac{u}{u-1}\,,
	\qquad
	\varphi_2=\dlog\frac{u}{u-1/x}\,,
\end{equation}
Here we have chosen differential forms that have logarithmic singularities exactly at the endpoints of the cycles. This type of differential form is called a canonical form associated to the cycle \cite{Arkani-Hamed:2017tmz}.\footnote{The notion of canonical form used here is a differential form and should not be confused with the differential equations in canonical form introduced in section~\ref{sec:diffeqs}.}
\end{ex}

In the following subsection, we use pairings of bases of twisted co/homology to construct  the period matrix and matrices of intersection numbers in twisted co/homology. The linear independence of a putative basis can be confirmed by checking the rank of these matrices. 
\begin{rem}One of the assumptions made in the mathematical treatment of intersection theory is that none of the exponents $\alpha_I$ is an integer, and moreover that their total sum is also not an integer. This assumption is not valid for typical Feynman integrals and their associated hypergeometric functions, but it is found in practice that physical results can be obtained by taking suitable limits of the generic case.
\end{rem}

\subsection{Pairings}

Now that we have introduced the twisted cohomology and homology groups associated to an integral $\int_\gamma \omega$, we consider functions that pair their elements through complex-valued bilinear maps. The first pairing is the familiar integral map, with a possible physical interpretation of a Feynman integral or a cut Feynman integral. Other pairings can be constructed between elements of the same group, giving intersection numbers of the twisted cohomology or twisted homology.

\paragraph{Integral pairings.}

The integral pairing brings us back to our starting point, namely the integral $\int_\gamma \omega$ obtained from the pairing of $\gamma$ and $\omega$.  Now that we have identified these objects as elements of twisted homology and cohomology groups, respectively, we can consider the full space obtained from this pairing, giving integrals related to the original one.

Let $\vec\gamma$ denote a basis of the (twisted) homology group, and let
 $\vec\varphi$ denote a basis of the (twisted) cohomology group. The integral pairings of the 
cycles $\gamma_l\in\vec\gamma$ with the forms
$\varphi_k\in\vec\varphi$ give entries of the so-called
\emph{period matrix}, 
\begin{equation}\label{eq:periodMat}
{ P}_{kl}(\vec\varphi,\vec\gamma)
=\gb{\gamma_l|\varphi_k}
	=\int_{\gamma_l}\Phi\varphi_k\,,
\end{equation}
where each row is associated to a differential form,  and each column is associated to a cycle.
Here we have also introduced the bra-ket notation of ref.~\cite{Mizera:2017rqa}, which will also be used for the other types of pairings. This notation should properly be accompanied by a subscript $\Phi$ to keep track of the twist, but since the twist remains constant throughout our discussion, we will suppress it.

In the context of Feynman integrals, each column of the period matrix comes from a basis of contours associated to generalised cuts, while each row of the period matrix comes from a basis of integrands associated to master integrals. The period matrix can in fact be interpreted as (the transpose of) a fundamental solution matrix for the system of differential equations satisfied by the master integrals, and is thus closely related to the Wronskian that was extensively used in \cref{sec:diffeqs}.

The period matrix ${P}$ is a square matrix whose dimension is given by the
dimension of the (co)homology group. 
It follows that any integral pairing of elements of the twisted homology and cohomology groups (with the same twist) can be written as a 
 linear combination of the elements of the period
matrix,
\begin{equation}
\label{eq:genericAsPeriods}
	\int_\gamma\omega
=	\sum_{k,l}c_{kl}\,{P}_{kl}(\vec\varphi,\vec\gamma)\,.
\end{equation}
The algebraic properties of any
integral of this type can then be studied
from the entries of the period matrix.

\begin{ex}
For the bases given in \cref{eq:2F1_contours,eq:canonicForm2F1}, we find the following entries of the period matrix:
\begin{align*}
P_{11}(\vec\varphi,\vec\gamma) &= 
\frac{\Gamma \left( a_0\epsilon \right) \Gamma
   \left( a_1\epsilon \right) }{\Gamma
   \left( a_0\epsilon+a_1\epsilon\right)} 
   \, _2F_1\left(
   a_0\epsilon ,- a_{1/x}\epsilon ;
   a_0\epsilon+a_1\epsilon ;x\right)\,, 
   \\
P_{12}(\vec\varphi,\vec\gamma) &= 
\frac{\pi  x^{-a_0 \epsilon } \csc \left(\pi  a_0
   \epsilon \right)  \Gamma \left( 1+
   a_{1/x}\epsilon \right)}{\Gamma \left(1-
   a_0\epsilon \right) \Gamma \left(1+  
  a_0\epsilon+a_{1/x}\epsilon\right)}
   \, _2F_1\left( 
   a_0\epsilon,1-  a_1\epsilon; 
   1 + a_0\epsilon+a_{1/x}\epsilon;\frac{1}{x}\right)\,,
    \\
P_{21}(\vec\varphi,\vec\gamma) &= 
\frac{\Gamma \left( a_0\epsilon \right) \Gamma
   \left( 1+ a_1\epsilon\right) }{\Gamma
   \left(1+ a_0\epsilon +a_1\epsilon \right)}
   \,
   _2F_1\left( a_0\epsilon ,1- 
   a_{1/x}\epsilon; 1+
   a_0\epsilon+a_1\epsilon;x\right) \,,
   \\
P_{22}(\vec\varphi,\vec\gamma) &= 
\frac{\pi  x^{-a_0 \epsilon } \csc \left(\pi  a_0
   \epsilon \right)  \Gamma \left(
   a_{1/x}\epsilon \right)}{\Gamma \left(1-
   a_0\epsilon \right) \Gamma 
   \left(a_0\epsilon +a_{1/x}\epsilon \right)}
\, _2F_1\left(
   a_0\epsilon ,- a_1\epsilon ; 
  a_0\epsilon+a_{1/x}\epsilon;\frac{1}{x}\right) \,.
\end{align*}
\end{ex}

\paragraph{Cohomology intersection numbers.}

A less obvious pairing can be constructed between
two differential forms. Consider
two bases $\vec\varphi$ and $\vec\psi$, not necessarily distinct, of the same twisted cohomology group.  We can then
compute \emph{cohomology intersection numbers} 
$\langle \varphi_i|\psi_j\rangle$ between these forms.
To be more precise, 
we must first construct a dual twisted cohomology group, which
is also generated by $\vec\psi$ but for which
the covariant differential is $\nabla_{\Phi^{-1}}$. For dimensionally regularised integrals, this dual operation corresponds to taking $\epsilon\to-\epsilon$
in $\Phi$. We can then pair generators 
$\langle\varphi_i|$ of the cohomology with elements
$|\psi_j\rangle$ of the dual cohomology \cite{AomotoKita,Mizera:2017rqa},
\begin{equation}\label{eq:intNumbDef}
	\langle \varphi_i|\psi_j\rangle
	=\frac{1}{(2\pi i)^2}\int_{X(\mathbb{C})}
	\varphi_i\wedge\iota_\Phi(\psi_j)\,,
\end{equation}
where $X(\mathbb{C})$ is as defined in \cref{eq:Xdef}, and
$\iota_\Phi$ is the map that associates to a form $\psi_j$
a form $\iota_\Phi(\psi_j)$ in the same cohomology class
but with compact support, so that the integral is well 
defined~\cite{Matsumoto1998,Mizera:2017rqa}. 
That is, the form $\iota_\Phi(\psi_j)$ vanishes in a neighborhood of the space $\Phi=0$ in ${\mathbb{P}^{n}(\mathbb{C})}$.
Unlike integral pairings, the cohomology intersection numbers are single-valued. 
The intersection
numbers of basis elements can  be arranged in the matrix
\begin{equation}\label{eq:C_matrix_def}
{C}_{kl}(\vec\varphi,\vec\psi)=\vev{\varphi_k|\psi_l}
\end{equation} 
which has the same dimensions as the period matrix ${P}$, and they similarly form a basis of all intersection numbers in the twisted cohomology.

The definition of the intersection numbers in \cref{eq:intNumbDef} 
is not always the most convenient for practical calculations,
even though it has recently been used in refs.~\cite{Caron-Huot:2021iev,Caron-Huot:2021xqj},
so alternative formulations have been found.
 One version pertains to the case
 where $n=1$, and the $\varphi_i$ and $\psi_j$ 
are $\dlog$-forms, i.e.\ wedge products of 1-forms $\dlog\alpha_k$, where the singularities of these forms are contained within the boundary variety $\Phi=0$.
In this case, 
a more explicit formula for the intersection numbers is given by
\cite{Mizera:2017rqa,Mastrolia:2018uzb,Mizera:2019gea}
\begin{equation}
	\label{eq:intRes}
	\vev{\varphi_i|\psi_j}
	=
	\sum_{u_p\in \mathcal{P}(\Phi)}
	\frac{\res_{u=u_p}\varphi_i\,\res_{u=u_p}\psi_j}
	{\res_{u=u_p}\dlog\Phi}\,,
\end{equation}
where 
$\mathcal{P}(\Phi)$ is the set of poles of $\dlog\Phi$.
This formula has been generalised to the case
where $n>1$ in ref.~\cite{Frellesvig:2019uqt}.

When the $\varphi_i$ and $\psi_j$ are not necessarily $\dlog$-forms,
other alternative formulas were proposed in ref.~\cite{Mizera:2017rqa}.
For $n=1$, and setting $u_1=u$, 
\begin{equation}
\label{eq:intnumbersell1}
\vev{\varphi_i|\psi_j} = \sum_{u^*} \left(
\frac{\partial^2\log\Phi}{\partial{u}^2}
\right)^{-1}
\left.
\widehat\varphi_i \,\widehat\psi_j
\right|_{u=u^*}\,,
\end{equation}
where the sum is over the critical points, i.e., the points $u^*$ satisfying 
$\dlog\Phi(u^*)=0$, and $\varphi_i=\widehat\varphi_i\, \rd u$ and similarly
for $\psi_j$.
In the case $n=2$, with $(u_1,u_2)=(u,v)$, 
\begin{equation}\label{eq:intnumbersell2}
\vev{\varphi_i|\psi_j} = \sum_{(u^*,v^*)} 
\left[\det
\left(
\begin{array}{cc}
 \frac{\partial^2\log\Phi}{\partial{u}^2} & 
 \frac{\partial^2\log\Phi}{\partial u \,\partial v} \\
 \frac{\partial^2\log\Phi}{\partial u \,\partial v} & 
 \frac{\partial^2\log\Phi}{\partial{v}^2}
\end{array}
\right)
\right]^{-1}
\left.
\widehat\varphi_i \,\widehat\psi_j
\right|_{(u,v)=(u^*,v^*)}\,,
\end{equation}
where  the sum extends over the critical points $(u^*,v^*)$ satisfying 
\[
\partial_u\log\Phi(u^*,v^*)=\partial_v\log\Phi(u^*,v^*)=0,\]
and these formulas generalise naturally to higher values of $n$. 
Reference~\cite{Weinzierl:2020xyy} presents a different algorithm using a Gr\"obner basis for generating a basis of twisted cohomology, with the feature that algebraic extensions such as square roots are avoided.

Cohomology intersection numbers are algebraic functions of the kinematic variables $x_j$ and exponents $\alpha_I$. 
For a $\dlog$ basis, the leading order in $\eps$ of the period matrix agrees with the matrix of intersection numbers \cite{Mastrolia:2018uzb}, which gives yet another method of computing these numbers.

\begin{ex}
In the example of the integral in \cref{eq:2f1int}, and the basis of \cref{eq:canonicForm2F1}, it is straightforward to see from 
\cref{eq:intRes} that the matrix of cohomology intersection numbers is
\begin{equation}\label{eq:matToDiag2F1}	
	{C}(\vec\varphi,\vec\varphi)
	=\vev{\vec\varphi|\vec\varphi}
	= 
	\begin{pmatrix}
	\dfrac{1}{a_0\epsilon}+\dfrac{1}{a_1\epsilon}  & \dfrac{1}{a_0\epsilon}
	\\[3mm]
	\dfrac{1}{a_0\epsilon}& \dfrac{1}{a_0\epsilon}+\dfrac{1}{a_{1/x}\epsilon} 
	\end{pmatrix}\,.
\end{equation}
\end{ex}

\paragraph{Homology intersection numbers.}

For completeness, we comment briefly on homology intersection numbers, although these 
have not yet been applied as widely in the context of physics.
In the definition of cohomology intersection numbers above, we saw that one of the differential forms needed to be given a compact support. The dual construction for twisted cycles is to use a locally finite chain group. Then the intersection numbers $[\gamma_k|\gamma_l]$ can be computed directly as geometric intersections of the cycles $\gamma_k$ and $\gamma_l$. We refer to section 2.3.3 of ref.~\cite{AomotoKita} for a fuller explanation and an illustrative example. Given a pair of bases $\vec\gamma, \vec\delta$ of twisted homology, the intersection numbers can be collected in matrix form, 
\begin{equation}
{H}_{kl}=[\gamma_k|\delta_l]\,.
\end{equation}
Here the bra elements $\sqbra{\gamma_l}$ are the twisted cycles described above, while the ket elements $\sqket{\delta_l}$ are generators of a locally finite version of twisted homology.

\paragraph{Twisted period relations.}

The pairings expressed in the period matrix co/homology intersection numbers give rise to the following quadratic relations \cite{ChoMatsumoto}:
\begin{align}\begin{split}\label{eq:twistedperiodrelations}
{P}_+\left({H}^{-1}\right)^T{P}_- &= {C}\,, \\
\left({P}_-\right)^T{C}^{-1}{P}_+ &= {H}\,. 
\end{split}\end{align}
The subscripts on the period matrix refer to two possible pairings: ${P}_{-kl}(\vec\varphi,\vec\gamma)=\gb{\varphi_k|\gamma_l}$ is the usual integral $\int_{\gamma_l}\Phi \varphi_k$, while ${P}_{+kl}(\vec\varphi,\vec\gamma)=\tgb{\varphi_k|\gamma_l}$ is a pairing between the corresponding compactified cocycles and locally finite cycles, which can simply be evaluated as the integral 
$\int_{\gamma_l}\Phi^{-1} \varphi_k$. For dimensionally regularised Feynman integrals, this duality is simply the transformation $\eps \to -\eps$.

As applied to the Gauss hypergeometric function, the twisted period relation can be expressed as the following quadratic relation:
\begin{align*}
{}_2F_1(\alpha,\beta;\gamma;x)
&{}_2F_1(1-\alpha,1-\beta;2-\gamma;x)
=\\
&\qquad{}_2F_1(\alpha+1-\gamma,\beta+1-\gamma;2-\gamma;x)
{}_2F_1(\gamma-\alpha,\gamma-\beta;\gamma;x)\,.
\end{align*}
Since twisted period relations are a generic feature of twisted cohomology and homology theories, they imply that there are  generically quadratic relations among dimensionally-regulated Feynman integrals. It would be interesting to explore these relations more systematically, and to relate them to the quadratic relations among Feynman integrals discussed in section~\ref{sec:other_rels}.

\subsection{Applications of intersection theory}

 This section presents a brief review of a few different applications of the intersection theory introduced above to the computation of Feynman integrals. We would also like to mention that intersection theory has found physical applications in the context of amplitudes in string theory and related theories, as initiated in ref.~\cite{Mizera:2017cqs,Mizera:2017rqa},
but we will not discuss those ideas here.

\subsubsection{Reduction and differential equations}

Intersection theory was applied to compute the coefficients in a reduction to master integrals in  refs.~\cite{Mastrolia:2018uzb,Frellesvig:2019kgj,Caron-Huot:2021iev,Caron-Huot:2021xqj}.
Suppose that the integral of interest can be represented as
\begin{equation}
\gb{\varphi|\gamma} = \int_\gamma \Phi\varphi\,.
\end{equation}
The reduction to master integrals can be represented as
\begin{equation} \label{eq:int-reduction}
\int_\gamma \Phi\varphi= \sum_i c_i \int_\gamma \Phi\varphi_i\,,
\end{equation}
where $\varphi_i$ is a basis of twisted cohomology, and the reduction coefficients are denoted by $c_i$.
This equation is fundamentally an operation on the cocycles, which we can write as
\begin{equation}
\bra{\varphi} = \sum_i c_i\bra{\varphi_i}\,.
\end{equation}
By projecting this equation onto the basis of cocycles, one finds that the reduction coefficients can be computed by cohomology intersection numbers:
\begin{equation}
c_i = \vev{\varphi|\varphi_j} \left({C}^{-1}\right)_{ji}\,.
\end{equation}
Thus the computation of cohomology intersection numbers can be substituted for reduction techniques. In particular,
this approach avoids having to solve large systems of equations, addressing one of the bottlenecks identified in \cref{sec:solIBPs}.
However, and despite much progress in the last few years, the difficulty associated with computing intersection numbers still makes the Laporta algorithm the method of choice in practical calculations.

When applied to the example of the ${}_2F_1$ integral, we find a version of the Gauss contiguous relations, which express the property that the space generated by integer shifts of the parameters $\alpha, \beta, \gamma$ in ${}_2F_1(\alpha,\beta,\gamma;x)$ is two-dimensional. In this case, it is equivalent to the statement that the twisted cohomology group is two-dimensional. 

Intersection theory can also provide the coefficients of the differential equations of \cref{eq:DEQ} \cite{Mastrolia:2018uzb}. If the differential operator is applied to the integral pairing, and if the integration contour does not depend on the kinematic variables, then we obtain
\begin{equation}\label{eq:int-diff}
\partial_{x_i} \int_\gamma \Phi\varphi_k
= \int_\gamma \Phi  \nabla_{\Phi,x_i} \varphi_k
\end{equation}
with the covariant derivative $\nabla_{\Phi,x_i}=\partial_{x_i} + \partial_{x_i}\log\Phi\wedge$, where here we take $\Phi=\mathcal{B}(z)^{\frac{D-K-1}{2}}$. The right-hand side can be expanded in master integrals as in \cref{eq:int-reduction}, and the same reasoning followed to obtain the coefficients in terms of intersection numbers. Thus we can identify the matrix of coefficients as
\beq
A_{x_i}(x,\eps) = \vev{\nabla_{\Phi,x_i}\varphi_j|\varphi_l}  \left(C^{-1}\right)_{lk}\,.
\eeq
In the Baikov representation, uncut (or less than maximally cut) Feynman integrals have singularities that are not regulated by the parameter $\epsilon$. These singularities can be treated with analytic regularisation \cite{speer1969}, which then allows the application of intersection theory \cite{Frellesvig:2019uqt,Frellesvig:2020qot}. A complementary approach is given in refs.~\cite{Caron-Huot:2021xqj,Caron-Huot:2021iev}, which focus on the cohomology that is dual to the cohomology of Feynman integrands, for the purpose of constructing intersection numbers as above. These papers find that the correct mathematical framework is that of 
 relative twisted co/homology  as introduced recently in \cite{matsumoto2019relative}, which extends the scope of intersection theory to allow relative boundaries.

\begin{rem}
The decomposition into a basis of master integrals with the help of intersection numbers has a dual application for integration cycles. Assume that we have fixed a basis of cut integrals, i.e., we have fixed a basis of integration cycles $\vec\gamma$. Then every other integration cycle $|\gamma]$ can be decomposed into this basis using the approach we have just described, but using the intersection numbers among cycles instead of intersections numbers among differential forms:
\beq
|\gamma] = \sum_i|\gamma_i]\,h_i\,,\textrm{~~~with~~~} h_i = \left(H^{-1}\right)_{ij}[\gamma_j|\gamma]\,.
\eeq
In this way we can decompose a given cut integral $\langle\gamma|\omega]$ into the basis of cut integrals $\langle\gamma_i|\omega]$.
\end{rem}

\subsubsection{Coaction}

MPLs are well known to be equipped with a coaction
\cite{GoncharovMixedTate,Goncharov:2005sla,brownmixedZ,Duhr:2012fh},
a mathematical operation that  is naturally compatible with the actions of differential operators and taking discontinuities across branch cuts. As such, it can be employed as a computational tool.

There is also a natural mathematical coaction on Feynman integrals~\cite{Brown:coaction, Panzer:2016snt,Brown:2015fyf}. 
Motivated by these observations,
it was conjectured that there exists a diagrammatic coaction on Feynman integrals which,
for the class of integrals admitting a Laurent expansion in $\epsilon$ in terms of MPLs, corresponds precisely to the coaction on MPLs. 
 This conjecture was expressed 
in refs.~\cite{Abreu:2017enx,Abreu:2017mtm,Abreu:2021vhb}, which proposed a general form for a coaction on integrals:\footnote{This form of a coaction is also believed to apply to the disk integrals in string tree-level amplitudes \cite{Schlotterer:2012ny,Drummond:2013vz}, as well as the more general genus-zero integrals in ref.~\cite{Britto:2021prf}.}
\begin{equation}\label{eq:masterformula}
	\Delta\left(\int_\gamma\omega\right)=
	\sum_{ij}c_{ij}\int_\gamma\omega_i\otimes
	\int_{\gamma_j}\omega\,,
\end{equation}
where the $c_{ij}$ are rational or algebraic coefficients, the $\{\omega_i\}$ are a basis of the (twisted) cohomology 
group associated with the integral on the left-hand side, and the $\{\gamma_j\}$ 
are are a basis of the corresponding homology group.\footnote{See also ref.~\cite{Kreimer:2020mwn} for a separate proposal of a diagrammatic coaction.}
It is straightforward to check that \cref{eq:masterformula} satisfies the algebraic properties of a coaction, such as coassociativity.
For one-loop integrals, the diagrammatic coaction was made fully explicit by providing bases of forms and contours for arbitrary one-loop integrals~\cite{Abreu:2017enx,Abreu:2017mtm}.

It was further conjectured in refs.~\cite{Abreu:2019wzk,Abreu:2019xep} that \cref{eq:masterformula} can be applied to generalized hypergeometric functions in their integral representations, of the type considered in this section, of which the prototypical example is the function ${}_2F_1$ that we have been following. The coaction operation proposed for hypergeometric functions is claimed to be compatible with the diagrammatic coaction, as well as to commute with taking the Laurent expansion in $\eps$, at least in the case where there is an expansion in terms of MPLs. 
For this class of functions, a rigorous (motivic) mathematical treatment has been developed in ref.~\cite{Brown:2019jng} and applied specifically to Lauricella $F_D$ functions. 
It is hoped that compatible coactions could be formulated for more general Feynman integrals, but this is more speculative, as it requires formulating coactions on the new classes of functions arising in the $\eps$ expansion.  This has been done for elliptic eMPLs in  ref.~\cite{Broedel:2018iwv}, and  an example of a corresponding diagrammatic coaction  is given in ref.~\cite{Hidding:2021xot}.

In the case where canonical forms $\Omega(\vec\gamma)$  can be constructed corresponding to the basis $\vec\gamma$,
the coaction on the period
matrix  is simply obtained by matrix multiplication:
\begin{equation}\label{eq:coactionPMat}
	\Delta P_{kl}(\vec\gamma,\vec\varphi)
	=
	\sum_{i,j}
	\left[C^{-1}\!\left(\Omega(\vec\gamma),\vec\varphi\right)\right]_{ij}
	P_{ki}(\vec\gamma,\vec\varphi)
	\otimes
	P_{jl}(\vec\gamma,\vec\varphi)\,.
\end{equation}
If one can choose generators of the (co)homology
such that $C\!\left(\Omega(\vec\gamma),\vec\varphi\right)=\delta_{ij}$,
then the coaction takes a particularly simple form:
\begin{equation}\label{eq:coactionDiagBasis}
	\Delta P_{kl}=
	\sum_{i}
	P_{ki}
	\otimes
	P_{il}\,.
\end{equation}

\begin{ex}
Throughout this section, we have followed the example of the hypergeometric function ${}_2F_1$. Any choice of bases will give an expression for the coaction using the general formula \cref{eq:coactionPMat}. But let us see how to arrive at a particularly elegant expression. 

For the basis of contours, let us take  $\gamma_1=[0,1]$ from the usual Euler integral, and $\gamma_2=[1/x,\infty)$.  One benefit of using $\gamma_2$ is that its pairing with the general integrand 
can be recognised as a ${}_2F_1$ function of $x$ after a simple change of integration variable from $u$ to $v=1/(xu)$.

It may be desirable to choose a second basis of differential forms $\varphi_i$, such that 
the matrix of intersection numbers ${C}(\vec\varphi,\vec\psi)$ 
has a minimum number of nonvanishing off-diagonal elements.
Indeed,
${C}$ will be diagonal if each $\varphi_i$ is taken to be a $\dlog$-form whose
singularities overlap with the boundary components of  $\gamma_j$ if and only 
if $i=j$. 
In this case, our chosen cycles $\gamma_1$ and $\gamma_2$ have no common boundaries. The choice
$\vec\varphi=(\varphi_1,\varphi_2)$ with
\begin{equation}\label{eq:cohom2f1}
\varphi_1 = \frac{a_0 a_1}{a_0+a_1}\, \eps\, \frac{du}{u(1-u)}, \qquad 
\varphi_2 = \frac{(a_0+a_1+a_{1/x})a_{1/x}}{a_0+a_1}\,\eps\,\frac{x\,du}{1-xu},
\end{equation}
uses $\dlog$ forms normalised so that the  matrix of cohomology intersection numbers is the identity.

With the bases of cycles $\vec\gamma$ and forms $\vec\varphi$, where $C\big(\Omega(\vec\gamma);\vec\varphi;\Phi\big)$ is the identity matrix, the coaction takes the form
\begin{eqnarray}
\label{eq:coaction2f1}
\Delta \left(\int_{\gamma_{1}} \Phi\varphi \right)
= \int_{\gamma_{1}} \Phi\varphi_1 
\otimes \int_{\gamma_{1}} \Phi\varphi
+ \int_{\gamma_{1}} \Phi\varphi_2 
\otimes \int_{\gamma_{2}} \Phi\varphi\,,
\end{eqnarray}
If we evaluate the integrals explicitly, and also include the coaction on the beta-function prefactors in \cref{eq:2f1def}, we find 
\begin{equation*}
\begin{split}
\Delta\Big({}_2F_1(\alpha,\beta;&\gamma;x)\Big) =
{}_2F_1(a\eps,b\eps;c\eps;x) \otimes {}_2F_1(\alpha,\beta;\gamma;x) \\
\label{eq:coaction2F1}
&- x\frac{b\eps(c\eps-b\eps)}{c\eps(1+c\eps)}\,
{}_2F_1(1+a\eps,1+b\eps;2+c\eps;x) \\
& \otimes 
e^{-i\pi(\alpha+\beta-\gamma)}x^{1-\gamma}
\frac{B(1-\beta,1+\beta-\gamma)}{B(\alpha,\gamma-\alpha)}
{}_2F_1\left(1+\alpha-\gamma,1+\beta-\gamma;2-\gamma;x\right),
\end{split}
\end{equation*}
where $\alpha=n_\alpha+a\eps$, $\beta=n_\beta+b\eps$ and
$\gamma=n_\gamma+c\eps$, and the coaction has also been applied to the beta function prefactors.
The hypergeometric functions in the last entry, ${}_2F_1(\alpha,\beta;\gamma;x)$  and $x^{1-\gamma}{}_2F_1\left(1+\alpha-\gamma,1+\beta-\gamma;2-\gamma;x\right)$, may be recognised as two independent solutions of Euler's hypergeometric differential equation.
\end{ex}

\section*{Acknowledgments}
We would like to thank Matteo Becchetti, Johannes Br\"odel, Kilian B\"onisch, Francis Brown, Fabian Fischbach, Einan Gardi, Harald Ita, Albrecht Klemm, Nils Matthes, James Matthew, Andrew McLeod, Sebastian Mizera, Christoph Nega, Ben Page, Lorenzo Tancredi, and Mao Zeng for fruitful discussions on some of the topics presented in this review.
This work  was supported by the European Union's Horizon 2020 research and innovation programme under the Marie Sk\l{}odowska-Curie grant agreement No.~764850 {\it ``\href{https://sagex.org}{SAGEX}''}.

\section*{References}
\bibliography{References}

\end{document}